\documentclass[11pt]{article}

\usepackage[a4paper]{geometry}
\newgeometry{margin=3cm}

\usepackage[english]{babel}
\usepackage[applemac]{inputenc}       
\usepackage[T1]{fontenc}

\usepackage{amssymb}
\usepackage{amsmath}
\usepackage{amsthm}
\usepackage{mathtools}

\usepackage{enumerate}

\usepackage{dsfont}

\usepackage{graphicx}

\usepackage{tikz}
\usetikzlibrary{calc}
\usetikzlibrary{arrows}
\usetikzlibrary{positioning}
\usetikzlibrary{fit}

\newcommand{\abs}[1]{\left\lvert#1\right\rvert}

\theoremstyle{plain}
\newtheorem{thm}{Theorem}[section]

\newtheorem{lem}[thm]{Lemma}
\newtheorem{prop}[thm]{Proposition}

\theoremstyle{definition}
\newtheorem{defn}[thm]{Definition}
\newtheorem{rem}[thm]{Remark}

\newtheorem{exmp}[thm]{Example}

\newcommand{\ee}{{\mathrm e}}
\newcommand{\ii}{{\mathrm i}}
\newcommand{\dd}{{\mathrm d}}
\DeclareMathOperator{\im}{im}
\DeclareMathOperator{\tr}{tr}
\DeclareMathOperator{\rk}{rk}
\DeclareMathOperator{\ran}{ran}

\let\Im\relax

\DeclareMathOperator{\Im}{Im}
\newcommand{\bvec}[1]{\vec{#1}}
\renewcommand{\bvec}[1]{\boldsymbol{#1}}

\DeclarePairedDelimiter\ket{\lvert}{\rangle}
\DeclarePairedDelimiterX\braket[2]{\langle}{\rangle}{#1 , #2}

\usepackage[pdftex,%
colorlinks=true,linkcolor=blue,citecolor=blue,urlcolor=blue
]
{hyperref}

\begin{document}
	
\title{Topology in shallow-water waves:  \\ a violation of bulk-edge correspondence}

\author{Gian Michele Graf, Hansueli Jud and Cl\'ement Tauber\thanks{tauberc@phys.ethz.ch} \\ \bigskip {\small Institute for Theoretical Physics, ETH Z\"{u}rich, Wolfgang-Pauli-Str. 27, 8093 Z\"{u}rich, Switzerland}}

\date{\today}

\maketitle

\vspace{-0.5cm}
\begin{abstract}
	We study the two-dimensional rotating shallow-water model describing Earth's oceanic layers. It is formally analogue to a Schr\"odinger equation where the tools from topological insulators are relevant. Once regularized at small scale by an odd-viscous term, such a model has a well-defined bulk topological index. However, in presence of a sharp boundary, the number of edge modes depends on the boundary condition, showing an explicit violation of the bulk-edge correspondence. We study a continuous family of boundary conditions with a rich phase diagram, and explain the origin of this mismatch. Our approach relies on scattering theory and Levinson's theorem. The latter does not apply at infinite momentum because of the analytic structure of the scattering amplitude there, ultimately responsible for the violation.
\end{abstract}

\section{Introduction}

Concepts developed to describe topological insulators can be applied far beyond their original context of the quantum Hall effect, or more generally, that of solid state physics. They are actually relevant to classical wave phenomena occurring in various fields such as optics \cite{RaghuHaldane08,DeNittisLein17}, acoustics \cite{Perietal19} or even fluid dynamics \cite{DelplaceMarstonVenaille17}, as soon as the partial differential equations ruling the system are formally equivalent to a Schr\"odinger equation and to the extent that they both engender analogous geometric structures. 

A central concept in topological insulators is the bulk-edge correspondence \cite{Halperin82}. It states that, when an infinite and gapped system -- the bulk -- admits a topological index, the latter predicts the number of chiral modes appearing at the edge of a sample with a boundary. More precisely such modes are counted by a topological edge index, which coincides with the bulk one. The correspondence was established in a wide range of settings, starting with \cite{Halperin82,SchulzBaldesKellendonkRichter00} at different levels of rigour and followed by \cite{GrafPorta13,Avilaetal13,ProdanSchulzBaldes16,EssinGurarie11} and others by including refinements, such as due to symmetries, and by using various methods.

In this paper we study a quasi two-dimensional, rotating and classical fluid called the shallow-water model. Such a model describes certain oceanic and atmospheric layers on Earth, and explains the presence of a large structure, called Kelvin equatorial wave, propagating near the equator in the Pacific ocean. Such a propagation is always from West to East with a remarkable stability. Ref.~\cite{DelplaceMarstonVenaille17} first provided an interpretation of the Kelvin wave as a topological mode at the interface between the two hemispheres. By changing sign at the equator, the Coriolis force is analogue to a magnetic field in the quantum Hall effect as already noticed earlier in~\cite{FroehlichStuderThiran95}. Later, it was also realized that each hemisphere has a well-defined bulk index -- Chern number -- after adding an odd-viscous term which provides a small-scale regularization for this continuous model \cite{Souslovetal19,TauberDelplaceVenaille19}.

The most striking feature of this model is a violation of the bulk-edge correspondence: The number of edge modes for a sample with a sharp boundary, like a coast, depends on the boundary condition and hence does not always match with the associated bulk index. Such a mismatch was conjectured in \cite{TauberDelplaceVenaille19bis} for some boundary condition. The main result of this paper is to prove it for a continuous family of conditions and to explain the cause of such a violation. 

Bulk-edge correspondence was proved in a very general setting for two-dimensional discrete systems with translation invariance \cite{GrafPorta13}. One approach relies on scattering theory, that studies how plane waves that come from the bulk are reflected at the boundary. The associated scattering amplitude encodes for the number of edge modes merging with a band edge in accordance with a variant of Levinson's theorem. Ultimately, it relies on the analytic continuation of the Bloch variety. The main difference here is that our model is continuous, so that the momentum as well as the Hamiltonian are not bounded. Even though the bulk picture is properly compactified, the analytic structure of the scattering amplitude at infinite momentum is exceptional and leads to two alternatives to Levinson's scenario. Both fail in counting the asymptotic number of edge modes, which clarifies the anomaly in the bulk-edge correspondence. To our knowledge, this is one of the rare cases of Levinson's theorem where the scattering amplitude is not trivial at infinity \cite{KellendonkPankrashkinRichard11}.

It is finally worthwhile to mention another approach to deal with a topological index for continuous models. A way to regularize the edge problem is to consider a smooth boundary or interface potential, gluing two samples with different bulk indices. The bulk-edge correspondence is usually satisfied in that case \cite{Feffermanetal16,Bal19,Drouot19}. 

The paper is organized as follows. Sect.~\ref{sec:main} describes the model from its physical origin to its topological bulk and edge features, and states the main result in terms of scattering theory. Sect.~\ref{sec:proofs} is devoted to the proofs and also provides further details about the mismatch. The appendices generalize some of the results beyond the particular choices that are made in the main text.

\section{Shallow-water model and its topology \label{sec:main} }

\subsection{The linearized, rotating and odd-viscous shallow-water model}

The shallow-water model describes a thin layer of fluid between a flat bottom and a free surface  \cite{Vallis17}. It has three degrees of freedom: the vertical height of the surface $\eta(x,y,t)$ and a horizontal two-component velocity field $u(x,y,t)$, $v(x,y,t)$. They are ruled by a system of partial differential equations:
\begin{subequations}\label{eq:ShallowWaterModel}
	\begin{align}
		\label{eq:ShallowWaterModel_1}\partial_t \eta &= - \partial_x u - \partial_y v,\\
		\label{eq:ShallowWaterModel_2}\partial_t u &=-\partial_x \eta- \left(f +\nu  \nabla^2 \right)v, \\
		\label{eq:ShallowWaterModel_3}\partial_t v &=-\partial_y \eta+ \left(f +\nu  \nabla^2 \right)u. 
	\end{align}
\end{subequations}
This model is derived from the three-dimensional Euler equations for an incompressible and homogeneous fluid. Equation \eqref{eq:ShallowWaterModel_1} comes from mass conservation, whereas \eqref{eq:ShallowWaterModel_2} and \eqref{eq:ShallowWaterModel_3} come from horizontal momentum conservation. The main assumption is that the typical wavelength of the fluid is much larger than its height. This allows to neglect vertical acceleration and implies hydrostatic pressure, leading to the $-\partial_x \eta$ and $-\partial_y \eta$ terms (gravity $g$ has been rescaled to 1). Moreover $u$ and $v$ are the depth-averaged horizontal components of the three dimensional velocity field. The system above is then obtained by linearizing the problem by looking at small fluctuations around a layer of fluid at rest.

When the fluid layer is the ocean, one takes into account Earth's rotation through the Coriolis acceleration $f (v,-u)$ where $f$ depends on the latitude.  It is positive (resp. negative) in the northern (resp. southern) hemisphere, and vanishes at the equator. Finally, the term $\nu \nabla^2 (v,-u)$ is called odd viscosity and comes from the antisymmetric part of the viscosity tensor, meant as a map between symmetric tensors of rank 2 \cite{Avron98}. This exotic term is non dissipative and allowed in dimension two if time reversal symmetry is broken. In the context of geophysical fluids this effect is not manifest, but it appears in some active liquids \cite{BanerjeeSouslovAbanovVitelli17,Souslovetal19}, and also in the quantum Hall effect, where it is called Hall viscosity. In the following $\nu$ is some positive and arbitrarily small parameter that regularizes the problem at small scales \cite{TauberDelplaceVenaille19,TauberDelplaceVenaille19bis}.

\subsection{Topology in the bulk \label{sec:bulk}}

The topology of shallow-water waves is revealed by studying their internal structure \cite{DelplaceMarstonVenaille17,Souslovetal19,TauberDelplaceVenaille19}. We approximate some local region on Earth by its tangent plane, so that $(x,y) \in \mathbb R^2$ and $f>0$ is a constant. We also require $\nu < 1/4f$ so as to streamline some computations below. The previous system \eqref{eq:ShallowWaterModel} is analogous to a Schr\"odinger equation with
\begin{equation}\label{eq:ShallowWater_Schrodinger} 
	\ii \partial_t \psi = \mathcal H \psi, \qquad \psi = \begin{pmatrix}
		\eta \\ u\\ v
	\end{pmatrix}, \qquad \mathcal H = \begin{pmatrix}
		0 & p_x & p_y \\ p_x & 0 & -\ii (f-\nu  \bvec{p}^2) \\ p_y & \ii (f-\nu  \bvec{p}^2) & 0
	\end{pmatrix},
\end{equation}
where $p_x = -\ii \partial_x$, $p_y = -\ii \partial_y$  and $\bvec{p}^2 = p_x^2+p_y^2$. $\mathcal H$ is a self-adjoint operator on $L^2(\mathbb R^2)^{\otimes 3}$ with domains $H^1(\mathbb R^2)\oplus H^2(\mathbb R^2) \oplus H^2(\mathbb R^2)$. It is also translation invariant so that the stationary solutions are given by the normal modes $\psi := \widehat \psi(k_x,k_y,\omega) \, \ee^{\ii(k_x x + k_y y-\omega t )}$ with momentum $\bvec{k} =(k_x,k_y) \in \mathbb R^2$ and frequency $\omega \in \mathbb R$, leading to the eigenvalue problem
\begin{equation}\label{eq:ShallowWater_bulk}
 	H \widehat \psi = \omega \widehat \psi, \qquad \widehat \psi = \begin{pmatrix}
		\hat \eta \\ \hat u\\ \hat v
	\end{pmatrix} \qquad  H(\bvec{k}) = \begin{pmatrix}
		0 & k_x & k_y \\ k_x & 0 & -\ii (f-\nu  \bvec{k}^2) \\ k_y & \ii (f-\nu  \bvec{k}^2) & 0
	\end{pmatrix},
\end{equation}
with $\bvec{k}^2=k_x^2+k_y^2$ and $H(\bvec{k})$ a Hermitian matrix. The system \eqref{eq:ShallowWater_bulk} admits three frequency bands:
\begin{equation}\label{eq:bulk_bands}
	\omega_\pm(\bvec{k}) = \pm \sqrt{\bvec{k}^2 + \left(f-\nu \bvec{k}^2\right)^2}, \qquad \omega_0(\bvec{k}) = 0,
\end{equation}
which are separated by two gaps of size $f$. In contrast to models that are periodic with respect to a lattice, such as (discrete) tight binding models, here momentum space is unbounded. As we shall see shortly, it is however appropriate to compactify it. Each band may then carry a non-trivial topology, characterized by a Chern number. The latter encodes the obstruction of finding a global eigensection that is non-vanishing and regular for all $\bvec{k} \in \mathbb R^2$ \cite{TauberDelplaceVenaille19}.

The Hamiltonian \eqref{eq:ShallowWater_bulk} can be rewritten as $H = \vec{d} \cdot\vec{S}$ where $\vec d = (k_x,k_y,f-\nu \bvec{k}^2)$ and
\begin{equation}
	S_1 = \begin{pmatrix}
		0 & 1& 0 \\ 1& 0 & 0\\ 0&0&0
	\end{pmatrix}, \qquad S_2 = \begin{pmatrix}
		0 & 0& 1 \\ 0& 0 & 0\\ 1&0&0
	\end{pmatrix}, \qquad S_3 = \begin{pmatrix}
		0 & 0& 0 \\ 0& 0 & - \ii\\ 0&\ii&0
	\end{pmatrix},
\end{equation}
is an irreducible spin 1 representation. $H$ shares its eigenprojection with the flat Hamiltonian $ H'=\vec e \cdot \vec S$ where $\vec e = \vec d/|\vec d|$. It reads 
\begin{equation}
	P_\pm = \dfrac{1}{2}\big( (\vec e \cdot \vec S)^2 \pm \vec e \cdot \vec S \big), \qquad P_0 = \mathds 1 - (\vec e \cdot \vec S)^2.
\end{equation}
We note that $\vec{e}=\vec{e}(\bvec{k})$ is convergent for $k\to \infty$, and so are $P_\pm$ and $P_0$; in fact $\vec{e}\to(0,0,-1)$ by $\nu>0$. Consequently, the Chern number 
\begin{equation}\label{eq:chern_tr}
	C(P) = \dfrac{1}{2 \pi \ii} \int_{\mathbb R^2} \dd k_x \dd k_y \, \tr( P [\partial_{k_x} P, \, \partial_{k_y}P])
\end{equation}
is a well-defined topological invariant. Indeed the momentum plane can be compactified to the 2-sphere $S^2$, so that the Berry curvature on the r.h.s is eventually computed on a closed manifold. If $\nu =0$ instead, $\vec{e}$ has a circle worth of accumulation points as $k\to \infty$. The r.h.s. may still be  finite but it would not be a Chern number.

\begin{prop}\label{prop:chern_numbers}
	Let $M$ be a compact two-dimensional manifold without boundary, $\vec e : M \to S^2$ and $H= \vec e \cdot \vec S$ with $\vec S$ an irreducible spin $s$ representation. Let $P_m$ be the eigenprojection of $H$ for the eigenvalue $m\in\{-s, -s+1, \ldots, s-1, s\}$. If the $2s+1$ bands of $H$ are non-degenerate over $M$ then one has
	\begin{equation}\label{eq:chern_volume}
		C(P_m) =   \dfrac{m}{2\pi} \int_M (\vec e)^*w, \qquad (\vec e)^*w = \vec e \cdot ( \partial_1 \vec e \wedge \partial_2 \vec e)\, \dd x_1 \dd x_2\,,
	\end{equation}
	where $w$ is the volume form on $S^2$. In particular if $\vec e$ wraps exactly once around the sphere then $C(P_m)=2 m$.
\end{prop}
The proof is given in App.~\ref{app:chern_spin}. In the case of \eqref{eq:ShallowWater_bulk} the map $\vec e : \mathbb R^2 \cong S^2 \mapsto S^2$ wraps exactly once around the sphere when $f$ and $\nu$ have the same sign. Given that $s=1$ we infer $C_\pm =\pm 2$ and $C_0=0$. In the rest of the article we will assume $f$, $\nu>0$.

\begin{rem}
In absence of regularization ($\nu=0$) the image of $\vec e$ has boundaries, as it covers only half the sphere. Therefore the r.h.s of \eqref{eq:chern_volume} happens to be an integer, $\pm 1$, but not a topological invariant. Indeed, $\vec e$ can be continuously deformed to the constant map $\vec e_0 = (0,0,1)$ that has zero Chern number. The integer value obtained for $\nu = 0$ is due to the normalization, especially to $s=1$. The analogue with a spin $s=1/2$ (Dirac Hamiltonian) leads to $\pm 1/2$ without regularization.
\end{rem}

\subsection{Edge modes \label{sec:edge}}

The non-trivial topology in the bulk should be manifest by the presence of edge modes in a sample with a boundary, according to the bulk-edge correspondence \cite{Halperin82,Hatsugai93}. We thus study the shallow-water problem in the upper half-plane $(x,y) \in \mathbb R \times \mathbb R_+$ with a horizontal boundary at $y=0$, where we impose the following condition:
\begin{equation}\label{eq:boundary_condition}
	v|_{y=0} = 0, \qquad (\partial_x u + a \partial_y v)|_{y=0} = 0,
\end{equation}
for some $a \in \mathbb R$. The first constraint means that the velocity at the boundary has no normal component. The second one is less easy to interpret. It was studied in \cite{TauberDelplaceVenaille19bis} for $a=\pm 1$. When $a=1$ the constraint implies by \eqref{eq:ShallowWaterModel_1} that $\eta$ is fixed to a constant value at the boundary.
For $a=-1$ the (odd-viscous) stress tensor $\sigma_{xy} = \nu (-\partial_x u + \partial_y v)$ vanishes, so that there is no shear at the boundary. Here we shall consider the entire family of conditions for $a \in \mathbb R$, in order to study the transition between different regimes. 

\begin{prop}
	For any $a\in\mathbb{R}$, the Hamiltonian $\mathcal H$ in \eqref{eq:ShallowWater_Schrodinger} is self-adjoint on $L^2\left(\mathbb{R}\times\mathbb{R}^+\right)^{\otimes 3}$ when equipped with boundary conditions \eqref{eq:boundary_condition}; more precisely its domain is the subspace of the Sobolev space $H^1\oplus H^2 \oplus H^2$ ($H^k=H^k\left(\mathbb{R}\times\mathbb{R}^+\right)$) defined by them.
\end{prop}
The proof is provided in App.~\ref{app:selfadjoint_bc} where we classify all self-adjoint boundary conditions. A rough count would suggest that a second order system with three unknowns, such as \eqref{eq:ShallowWaterModel}, would require three boundary conditions so as to ensure self-adjointness. However \eqref{eq:ShallowWaterModel_1} is first order, lowering the count by one.

The problem is translation invariant in the $x$-direction, so that the stationary solutions are given by the normal modes $\psi = \widetilde \psi \, \ee^{\ii(k_x x-\omega t)}$ with momentum $k_x \in \mathbb R$, frequency $\omega \in \mathbb R$ and $\widetilde \psi (y; k_x, \omega) =: (\tilde \eta, \tilde u, \tilde v)$. The system \eqref{eq:ShallowWaterModel} becomes a system of ordinary differential equations
\begin{subequations}\label{eq:ShallowWater_edge}
	\begin{align}
	\label{eq:ShallowWater_edge_1}\ii \omega \tilde \eta &= \ii k_x \tilde u + \partial_y \tilde v,\\
	\label{eq:ShallowWater_edge_2}\ii \omega \tilde u &=\ii k_x \tilde \eta + \left(f -\nu  k_x^2 \right)\tilde v  + \nu \partial_{yy} \tilde v,\\
	\label{eq:ShallowWater_edge_3}\ii \omega \tilde  v &=\partial_y \tilde \eta - \left(f -\nu  k_x^2 \right) \tilde u  - \nu \partial_{yy} \tilde u,
	\end{align}
\end{subequations}
that is exactly solvable for each value of the parameters $k_x,\,\omega$ and $a$.  

\begin{figure}[htb]
	\centering
	\includegraphics[scale=0.5]{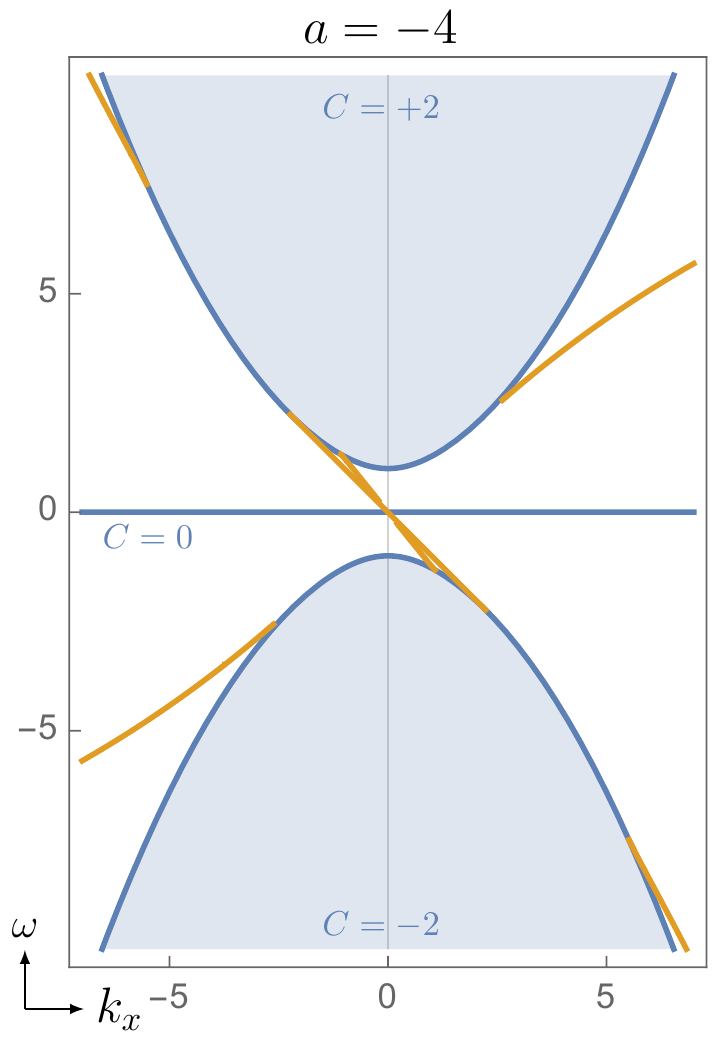}\hspace{0.25cm}
	\includegraphics[scale=0.5]{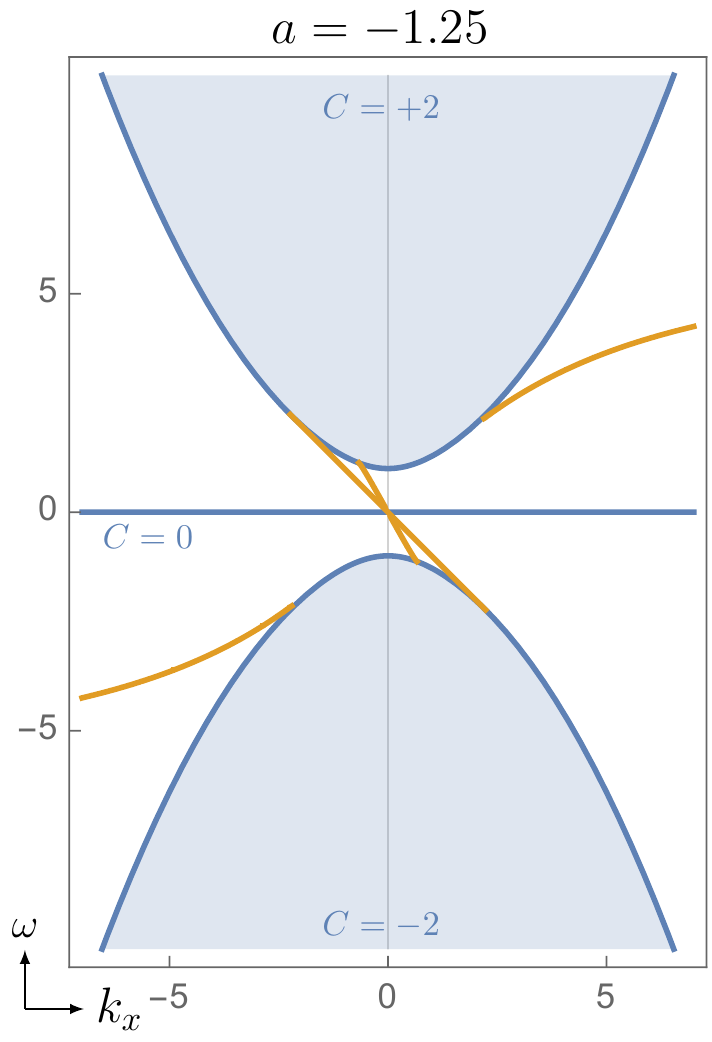}\hspace{0.25cm}
	\includegraphics[scale=0.5]{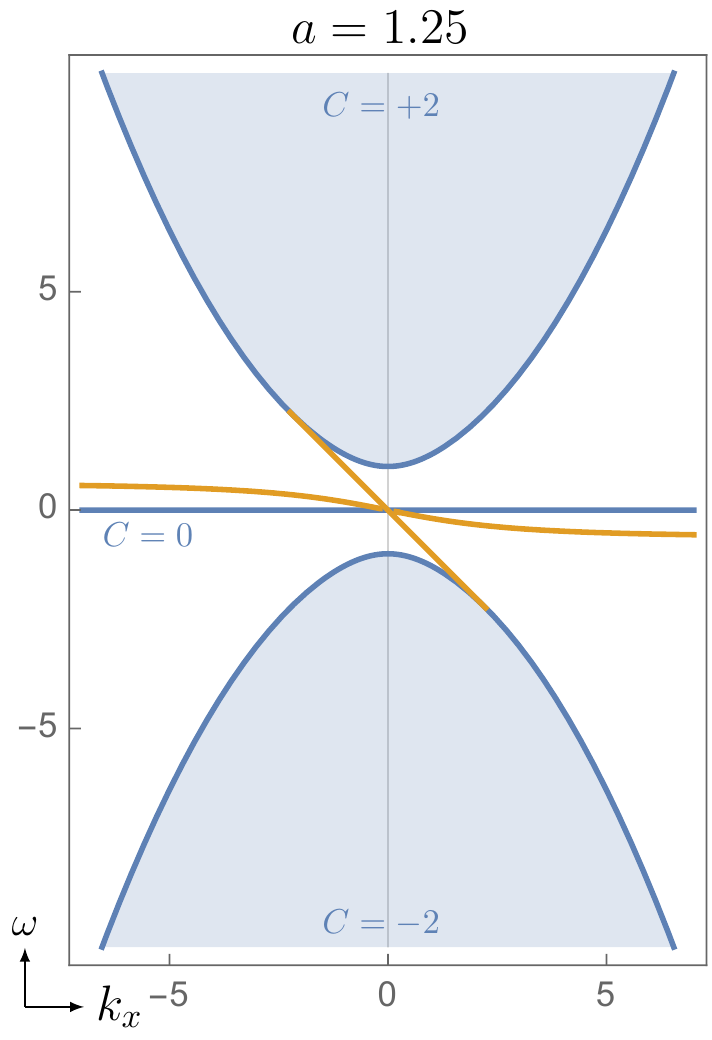}\hspace{0.25cm}
	\includegraphics[scale=0.5]{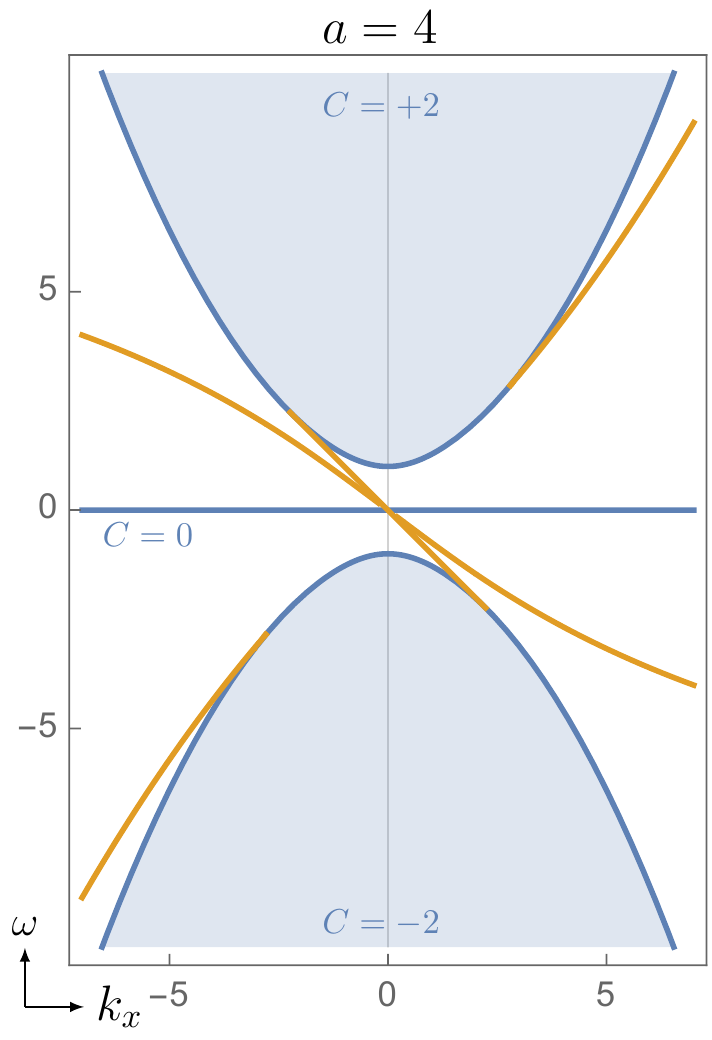}	
	\caption{Spectrum of the edge problem \eqref{eq:ShallowWater_edge} with boundary condition \eqref{eq:boundary_condition} for $f=1$, $\nu=0.2$ and four values of $a$. Shaded blue regions correspond to delocalized solutions extending over the half-plane. Yellow branches correspond to edge modes, localized near the boundary.\label{fig:edge_spectrum}}
\end{figure}
In Fig.~\ref{fig:edge_spectrum} the edge spectrum is plotted for different values of $a$, corresponding to the existence of solutions to \eqref{eq:ShallowWater_edge} that satisfy \eqref{eq:boundary_condition} and stay bounded when $y \rightarrow \infty$. The nature of the solutions depends on $k_x$ and $\omega$ and is of one of two types. For $|\omega| \geq \omega_+(k_x,0)$ or $\omega=0$ the solutions (in blue) are delocalized in the upper half plane. The same blue region also corresponds by the way to bounded solutions in the whole plane, and is nothing but the projection of the surface generated by \eqref{eq:bulk_bands}. In the gaps between them, the yellow curves in the spectrum are edge modes that decay exponentially when $y \rightarrow \infty$. 

What is striking here is that the number of such modes changes with the choice of boundary condition. As we shall see this is in contradiction with the bulk-edge correspondence. Moreover in each case there are edge modes that saturate at $\omega = C^{st}$ as $|k_x|\rightarrow \infty$. Such modes are perfectly allowed when $k_x$ is unbounded and  have the physical interpretation of inertial waves in classical fluids \cite{Iga95}. Because of such branches we have to specify a consistent way to count edge modes.

\begin{defn}
	The number $n_\mathrm{b}$ of edge modes below a bulk band is the signed number of edge mode branches emerging $(+)$ or disappearing $(-)$ at the lower band limit, as $k_x$ increases. The number $n_\mathrm{a}$ of edge modes above a band  is counted likewise up to a global sign change.   
\end{defn}

In the following, we focus on the upper band only since the lower one is its symmetric and the middle one is trivial. By taking the diagrams of Fig.~\ref{fig:edge_spectrum} in the order of increasing $a$ one reads off $n_\mathrm{b}= 2,\, 3, \, 1,\,2$; moreover $n_\mathrm{a}=0$ in all cases, because the upper band has no upper edge. We defer any objections to this count and invoke bulk-edge correspondence, in the form of the Hatsugai relation. That principle, if accepted in the present context, would state
\begin{equation}\label{eq:Hatsugai_rel}
	C_+= n_\mathrm{b}- n_\mathrm{a}\,;
\end{equation}
yet it is violated, at least for some $a$, because only the l.h.s. is independent of it.

\begin{prop}\label{prop:nb_finite_kx}
The phase diagram of the total number $n_\mathrm{b}$ of edge modes below the upper band for $k_x$ in an arbitrary large but finite interval reads:
\begin{center}
	\begin{tikzpicture}
	\draw[-latex] (-4,0) -- (4,0)node[below]{$a$};
	\draw (0,0.2) -- (0,-0.2)node[below]{\small $0$};
	\draw (2,0.2) -- (2,-0.2)node[below]{\small $\sqrt 2$};
	\draw (-2,0.2) -- (-2,-0.2)node[below]{\small $-\sqrt 2$};
	\draw (-4.5,0.5) node{$n_\mathrm{b} = $};
	\draw (-3,0.5) node{$2$};
	\draw (-1,0.5) node{$3$};
	\draw (1,0.5) node{$1$};
	\draw (3,0.5) node{$2$};
 	\end{tikzpicture}
\end{center}
At the transition $a=\sqrt 2$, an edge mode branch existing for $a>\sqrt 2$ is repelled to $k_x = +\infty$ and vanishes from the spectrum for $a<\sqrt 2$, and likewise at $a =-\sqrt{2}$ and $k_x = - \infty$.
\end{prop}

A possible objection to the count is that the diagrams only cover a finite interval in $k_x$, thus missing some distant eigenvalue branches which, if included, could possibly yield $n_\mathrm{b}=2$ always. By the proposition this is explicitly not the case.
Another objection is that the definitions of $n_\mathrm{b}$ and $n_\mathrm{a}$ ought to be modified in situations like in the first and last diagrams, and more generally for $\abs{a}>\sqrt{2}$, since they feature one edge state that is asymptotic to the bulk spectrum at $k_x\rightarrow-\infty$ and $+\infty$, respectively. Since these cases are those for which \eqref{eq:Hatsugai_rel} holds true, a modification would not help.

The main goal of this paper is to explain such a mismatch in the bulk-edge correspondence. The proof of this proposition is a direct consequence of Thm.~\ref{thm:main} below. 

\subsection{Scattering theory} 

The scattering approach is an intermediate picture between bulk and edge that describes solutions of the edge problem as a superposition of bulk solutions, which are interpreted as scattering waves. It was used in \cite{GrafPorta13} to establish the bulk-edge correspondence for discrete models in solid state physics. Here we adapt it to continuous models where $k_x$ and $\omega$ are unbounded. 

In the upper half-plane, $k_y$ is not a good quantum number and the bulk normal mode $\psi = \widehat \psi \ee^{\ii ( k_x x  -\omega t)}$ from Sect.~\ref{sec:bulk} is a solution of the eigenvalue eq.~\eqref{eq:ShallowWater_edge} but does not satisfy the boundary condition~\eqref{eq:boundary_condition}. However for $\kappa >0$, $\ee^{\ii (k_x x - \kappa y-\omega t)}$ and $\ee^{\ii ( k_x x +\kappa y-\omega t )}$ can be seen as incoming and outgoing plane waves with respect to the boundary at $y=0$. Moreover they share the same frequency  $\omega_+(k_x,\kappa) = \omega_+(k_x,-\kappa)$. Actually, $\omega_+(k_x,k_y) = \omega_+(k_x,\kappa)$ admits two other solutions: $k_y=\kappa_\mathrm{ev},\, \kappa_\mathrm{div}$ that are purely imaginary. Explicitly,
\begin{equation}\label{eq:def_kappa_ev}
	\kappa_\mathrm{ev/div}(k_x,\kappa) = \pm \ii \sqrt{\kappa^2 + 2 k_x^2 + \dfrac{1-2\nu f}{\nu^2}} \in \pm \ii \mathbb R_+.
\end{equation}
The normal mode $\psi(k_x,\kappa_\mathrm{div})$ is divergent as $y \to \infty$ whereas $\psi(k_x,\kappa_\mathrm{ev})$ is evanescent away from the boundary. The former cannot be part of the solution to the boundary problem but the latter must be taken into account. 
\begin{defn}\label{def:scat_state} 
	For $k_x\in \mathbb R$ and $\kappa >0$ a \textbf{scattering state} is a solution $\psi_\mathrm{s} = \widetilde \psi_\mathrm{s} \ee^{\ii (k_x x-\omega t )}$ with $\omega=\omega_+(k_x,\kappa)$ of the form
	\begin{equation}\label{eq:scat_state}
		\widetilde \psi_\mathrm{s} = \psi_\mathrm{in} + \psi_\mathrm{out} + \psi_\mathrm{ev},
	\end{equation}
	satisfying the boundary condition~\eqref{eq:boundary_condition}, where the three terms correspond to bulk solutions of momenta $k_y =-\kappa\,,\kappa$, and $\kappa_{ev}$. The solution exists and is unique up to multiples.
\end{defn}
Given an open set $U_{out}\subset\mathbb{R}^2$, let $U_{in}\subset\mathbb{R}^2$ and $U_{ev}\subset\mathbb{R}\times\ii \mathbb{R}$ be the images under the maps $(k_x,\kappa)\mapsto(k_x,-\kappa)$ and $(k_x,\kappa)\mapsto (k_x,\kappa_{ev})$. Let $\psi_{in}=\psi_{in}(k_x,-\kappa)e^{-\ii \kappa y}$, $\psi_{out}=\psi_{out}(k_x,\kappa)e^{\ii \kappa y}$, $\psi_{ev}=\psi_{ev}(k_x,\kappa_{ev})e^{\ii \kappa_{ev} y}$ be a choice of solutions on $U_{in/out/ev}$, i.e. of sections that do not vanish anywhere in their domains. A unique solution~\eqref{eq:scat_state} is then singled out by requiring that the amplitude of $\psi_{in}$ be 1:
\begin{equation}\label{eq:scattering_solution}
	\widetilde\psi_s = \psi_{in} + S\psi_{out}+T\psi_{ev}
\end{equation}
with coefficients $S(k_x,\kappa)$, $T(k_x,\kappa)\in\mathbb{C}$
\begin{defn}\label{def:scat_ampl}
	$S(k_x,\kappa)$ is called the scattering amplitude for the chosen sections.
\end{defn}

\begin{rem}\label{rem:scat_state_uniqueness}
	The uniqueness of~\eqref{eq:scattering_solution} is conditioned on a choice of sections $\psi_{in}$, $\psi_{out}$, $\psi_{ev}$. The gauge freedom is to multiply any of them by a factor, $\widehat{\psi}\mapsto z\widehat{\psi}$, where $z\neq 0$ depends on $ k_x$, $\kappa$. For the discussion of scattering close to the threshold it will nevertheless be convenient to use a same section for $\psi_{in}$, $\psi_{out}$, that is moreover symmetric under $\kappa\mapsto -\kappa$. As a result, $\abs{S}=1$ (see Prop.~\ref{prop:S_U1}). The choice for $\psi_{ev}$ can be different, but we require all three to be non-vanishing and regular in a given neighborhood of interest.
\end{rem}
\begin{figure}[htb]
	\centering
	\includegraphics[scale=1]{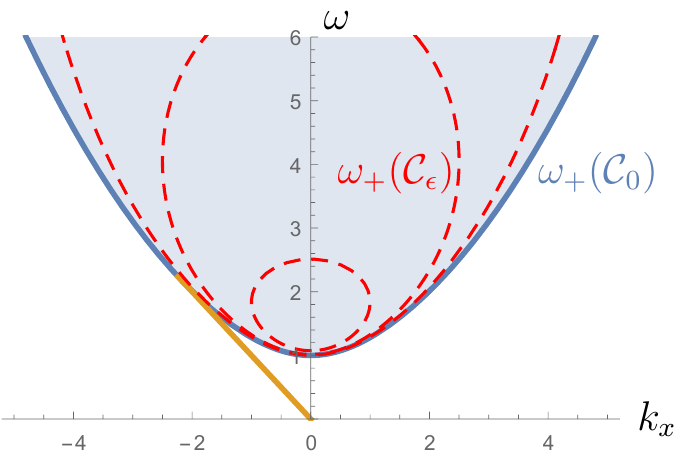}
	\caption{Image of $\mathcal C_\epsilon$ in the $(k_x,\omega)$-plane for several small values of $\epsilon$ (dashed red). As $\epsilon \to 0$, it approaches the bottom limit of the upper band. There, the argument of $S$ jumps by $2\pi$ when an edge mode branch disappears or emerges, according to the relative Levinson's theorem.\label{fig:dual_curve}}
\end{figure}
The scattering amplitude $S$ is on one hand a transition between bulk sections, and hence naturally related to the Chern number, and on the other hand by Levinson's theorem it is sensitive to the presence of edge modes when approaching the limit of the bulk band. To explore the bottom of the upper band, including $|k_x| = \infty$, we define the following dual variables
\begin{equation}\label{eq:dual_variables}
	k_x = \dfrac{\lambda_x}{\lambda_x^2+ \lambda_y^2}, \qquad k_y = \dfrac{-\lambda_y}{\lambda_x^2+ \lambda_y^2},
\end{equation}
for $\lambda_x,\lambda_y \in \mathbb R^2$. This is an orientation preserving change of variable that exchanges $0$ with $\infty$. For $\epsilon>0$ consider the following curve
\begin{equation}\label{eq:def_Cepsilon}
	\mathcal C_\epsilon = \left\{ \Big(k_x= \dfrac{\lambda_x}{\lambda_x^2+ \epsilon^2}, \kappa = \dfrac{\epsilon}{\lambda_x^2+\epsilon^2}+\epsilon\Big) | \lambda_x \in \check{\mathbb R} \right\}.
\end{equation}
This is a circle of center $(0,1/2\epsilon + \epsilon)$ and radius $1/2\epsilon$ in the $(k_x,\kappa)$-plane. One has $(k_x,\kappa) \to (0^\pm,\epsilon)$ as $\lambda_x \to \pm \infty$ and conversely $(k_x,\kappa) \to (0^\pm, 1/\epsilon+\epsilon)$ as $\lambda_x \to 0^\pm$. The choice of a reverse orientation for $\lambda_x \in \check{\mathbb R}$ implies that, in the limit $\epsilon \rightarrow 0$, $\mathcal C_\epsilon$ turns into the straight line $k_x \in \mathbb R, \kappa =0$, which corresponds to the bottom of the band $\omega_+(k_x,0)$. 

The image of $\mathcal C_\epsilon$ in the $(k_x,\omega)$ plane is plotted in Fig.~\ref{fig:dual_curve}. As $\epsilon \to 0$, it explores the bottom of the upper band, including $|k_x|= \infty$. Finally we separate the contributions near $k_x=0$ and $|k_x| = \infty$ as follows. For $\lambda_0>0$ we define
\begin{equation}\label{eq:def_Clambda0}
	\mathcal C_{\epsilon,\lambda_0} = \big\{ (k_x,\kappa) \in \mathcal C_\epsilon | \lambda_x \in [\lambda_0,-\lambda_0] \big\}, \qquad \mathcal C_{\epsilon,\lambda_0}^\perp = \mathcal C_\epsilon \setminus \mathcal C_{\epsilon,\lambda_0},
\end{equation}
so that $(0,0) \in \mathcal C_{0,\lambda_0}^\perp$ and $(\pm\infty,0) \in \mathcal C_{0,\lambda_0}$. The main result is then
\bigskip
\begin{thm}\label{thm:main}
	Let $a \in \mathbb R \setminus\{0,\pm\sqrt{2}\}$. The following statements hold:
	\begin{itemize}
		\item (Bulk-scattering correspondence) For all $\varepsilon>0$
		\begin{equation}\label{eq:BSC} 
			C_+ =  \dfrac{1}{2\pi \ii} \int_{\mathcal C_{\epsilon}} S^{-1} \dd S.
		\end{equation}
		\item (Relative Levinson's theorem) There exists $\lambda_0$ small enough such that $\forall \lambda, 0 < \lambda < \lambda_0$
		\begin{equation}\label{eq:Relative_Levinson}
		n_b = \lim\limits_{\epsilon \to 0} \dfrac{1}{2\pi \ii}  \int_{\mathcal C_{\epsilon,\lambda}^\perp} S^{-1} \dd S.
		\end{equation}
		\item (Violation of Levinson's theorem) There exists $\lambda_0$ small enough such that $\forall \lambda, 0 < \lambda < \lambda_0$,
		\begin{equation}\label{eq:Levinson_violation}
		\lim\limits_{\epsilon \to 0} \dfrac{1}{2\pi \ii}  \int_{\mathcal C_{\epsilon,\lambda}} S^{-1} \dd S = \left\lbrace \begin{array}{ll}
		0,& |a| > \sqrt{2}, \\
		\mathrm{sign}(a), & 0<|a|<\sqrt{2}.
		\end{array}\right.
		\end{equation}
		Moreover for $a>\sqrt{2}$ (resp. $<-\sqrt{2}$) there is an edge mode branch merging with the bulk band at $k_x = \infty$ (resp. $-\infty$). For $|a|<\sqrt{2}$ there are no edge modes in the neighborhood of the bulk band as $|k_x|\rightarrow \infty$.
	\end{itemize}
\end{thm}

Thus, the change of argument of the scattering amplitude along $\mathcal C_\epsilon$ coincides with the Chern number for every $a$, thanks to a non-trivial contribution \eqref{eq:Levinson_violation} at $|k_x|=\infty$ that compensates the missing or superfluous edge modes observed in Sect.~\ref{sec:edge}. Usually a jump in the argument of $S$ is associated to an edge mode branch disappearing or emerging from the bulk band limit \cite{GrafPorta13}. This statement is valid as long as $k_x$ is finite, leading to \eqref{eq:Relative_Levinson}. However, the opposite occurs at $|k_x|=\infty$: the argument of $S$ jumps while there is no edge mode branch merging in the spectrum, and conversely. This shows an explicit violation of Levinson's theorem and proves the mismatch in the number of edge modes from Prop.~\ref{prop:nb_finite_kx}.

The proofs of \eqref{eq:BSC} and \eqref{eq:Relative_Levinson} are done in \ref{sec:BSC_proof}, respectively \ref{sec:relative_Levinson}. They are adapted from \cite{GrafPorta13} where $k_x, k_y \in \mathbb T^2$, the two-dimensional Brillouin torus. The key point is to study the poles and zeros of $S$ after analytic continuation in $\kappa$. The main result of the paper, eq. \eqref{eq:Levinson_violation}, is proved in \ref{sec:Levinson_violation} by studying the singularity of such a complex continuation in  $\kappa$ as $|k_x| \rightarrow \infty$, eventually responsible for the violation of Levinson's theorem. Notice that $S$ depends on the choice of section appearing in \eqref{eq:scat_state}, but the theorem is true as long as such sections are regular in a neighborhood of $\mathcal C_\epsilon$. Moreover we can even use sections that have singularities near $\mathcal C_{\epsilon,\lambda}$ (resp. $\mathcal C_{\epsilon,\lambda}^\perp$) when dealing with \eqref{eq:Relative_Levinson} (resp. \eqref{eq:Levinson_violation}). Such singular sections have extra symmetries that actually simplify the proof.

\begin{rem}
	Such a violation is not always occurring in the shallow-water model. Indeed, the standard Dirichlet boundary condition:
	\begin{equation}\label{eq:Dirichlet_BC}
			u|_{y=0} = 0, \qquad v|_{y=0} = 0, 
	\end{equation}
	leads to $n_b=2=C_+$ with no argument jump of $S$ at $|k_x|=\infty$ and no asymptotic edge mode branch near the bottom of the upper bulk band as $|k_x| \rightarrow \infty$. Thus, unbounded parameters $(k_x,\omega)$ do not necessarily lead to a violation of Levinson's theorem. This non-anomalous case is detailed in App.~\ref{app:Dirichlet}.
\end{rem}

\section{Proofs\label{sec:proofs}}

\subsection{Sections, scattering amplitudes and bulk-scattering correspondence\label{sec:BSC_proof}}

Before proving Thm.~\ref{thm:main} we discuss the ambiguity in the definition of the scattering state and amplitude, due to the gauge freedom of  bulk eigensections. Additionally, we provide explicit expressions, that are used in the next sections. In the following we identify the compactified $k$-plane with the Riemann sphere $\mathbb C \cup \{\infty\} \cong S^2$ via $z =k_x+\ii k_y$. Since $C_+=2$, it is impossible to find a global bulk eigensection $\widehat \psi$ that is regular for all $z \in S^2$. We need at least two distinct ones, that are regular locally on two overlapping patches to cover the sphere. This leads to distinct scattering states and scattering amplitudes. It is readily verified that $H\widehat{\psi}^\infty=\omega_+\widehat{\psi}^\infty$, where
\begin{equation}\label{eq:section_infty}
	\widehat \psi^\infty(\bvec{k}) = \dfrac{1}{\sqrt{2}}\dfrac{1}{k_x- \ii k_y} \begin{pmatrix}
		\bvec{k}^2/\omega_+ \\ k_x-\ii k_y q \\ k_y+ \ii k_x q 
	\end{pmatrix}\,, \quad 
	q(\bvec{k}) := \tfrac{f-\nu \bvec{k}^2}{\omega_+}\,,\quad \omega_+=\omega_+(\bvec{k})\,.
\end{equation}
Notice that $q \rightarrow 1$ (resp. $-1$) as $k\rightarrow 0$ (resp. $\infty$). Thus \eqref{eq:section_infty} defines a section of the eigenbundle of $\omega_+$ that is smooth for all $z \in \mathbb C$, including $z=0$, but not at $\infty$, where it is singular and winds like $z/\abs{z}$. However $z=\infty$ belongs to the curve $\mathcal C_\epsilon$ as $\epsilon \rightarrow 0$, see \eqref{eq:def_Cepsilon}, so that $\widehat \psi^\infty$ cannot be used directly in the proof of Thm.~\ref{thm:main}. Instead we define for $\zeta=\zeta_x + \ii \zeta_y \in \mathbb C$
\begin{equation}\label{eq:section_zeta}
	\widehat \psi^\zeta = t^\zeta_\infty \widehat \psi^\infty, \qquad t^\zeta_\infty(z)=\frac{\bar{z}-\bar{\zeta}}{z-\zeta}
\end{equation}
which is regular for all $z \in S^2\setminus\{\zeta\}$, including $\infty$ and singular in $\zeta$. We shall mainly use $\zeta=\ii \zeta_y$ with $\zeta_y>0$ that is away from $\mathcal C_\epsilon$ for $\epsilon$ small enough, and $\zeta=0$ that is rotation invariant near $z=\infty$. According to Def.~\ref{def:scat_state}, the scattering state for each section $\widehat \psi^\zeta$, $\zeta \in S^2$, reads
\begin{equation}\label{eq:scat_explicit}
	\widetilde \psi_\mathrm{s}^\zeta := \psi_\mathrm{in}^\zeta + S_\zeta \psi_\mathrm{out}^\zeta + T_\zeta \psi_\mathrm{ev}^\infty,\\
\end{equation}
with 
\begin{equation}
	\psi_\mathrm{in}^\zeta = \widehat \psi^\zeta(k_x,-\kappa)\ee^{-\ii \kappa y}, \qquad  \psi_\mathrm{out}^\zeta = \widehat \psi^\zeta(k_x,\kappa)\ee^{\ii \kappa y}, \qquad \psi_\mathrm{ev}^\infty = \widehat \psi^\infty(k_x,\kappa_\mathrm{ev})\ee^{\ii \kappa_\mathrm{ev}y}.
\end{equation}
Notice that we have dropped the $\omega$-dependence since $\omega = \omega_+(k_x,\kappa)$. We recall that $\kappa >0$ in the definition of the scattering state, so that the choice of some $\zeta= \ii \zeta_y$ with $\zeta_y>0$ guarantees that $\psi_\mathrm{in}^\zeta$ is regular for  $k_x \in \mathbb R$ and  $\kappa >0$.  Moreover, $S_\zeta \psi_\mathrm{out}^\zeta$ is also regular in the whole upper half-plane, even though $\psi_\mathrm{out}^\zeta$ is singular at $\kappa=\zeta_y>0$. Finally, notice that $\kappa_\mathrm{ev}(k_x,\kappa) \to -\ii \nu^{-1}\sqrt{1-2\nu f} \neq 0$, $(k_x,\kappa) \to 0$ and $q(k_x,\kappa_\mathrm{ev}) \to 1$, $(k_x,\kappa)\to \infty$. Thus $\widehat \psi^\infty(k_x,\kappa_\mathrm{ev})$ is regular in the whole upper half-plane, including at $\infty$, and appears as a common and convenient choice of the evanescent part in \eqref{eq:scat_explicit} for all $\zeta \in S^2$. Note that because of Rem.~\ref{rem:scat_state_uniqueness} we can choose the phase of $\psi_\mathrm{ev}$ independently from the choice of ${\psi}_{in/out}$.

\begin{lem}
	Let $\zeta \in S^2$, $k_x \in \mathbb R$ and $\kappa >0$ so that $k_x + \ii \kappa \in S^2 \setminus\{\zeta\}$. Then
	\begin{equation}\label{eq:transition_S}
	S_\zeta(k_x,\kappa) = \dfrac{t^\zeta_\infty(k_x,-\kappa)}{t^\zeta_\infty(k_x,\kappa)} S_\infty(k_x,\kappa).
	\end{equation}
\end{lem}
\begin{proof}
The scattering amplitude can be computed explicitly from \eqref{eq:scat_explicit} and boundary condition \eqref{eq:boundary_condition}. For  $\zeta \in S^2$ we denote $u_\zeta$ and $v_\zeta$ the second and third component of $\widehat \psi^\zeta$, according to \eqref{eq:section_infty} or \eqref{eq:section_zeta}. We infer
\begin{equation}\label{eq:S_vs_g}
	S_\zeta(k_x,\kappa) = - \dfrac{g_\zeta(k_x,-\kappa)}{g_\zeta(k_x,\kappa)},\qquad T_\zeta(k_x,\kappa) = - \dfrac{h_\zeta(k_x,\kappa)}{g_\zeta(k_x,\kappa)},
\end{equation}
where
\begin{equation}\label{eq:def_g_zeta}
	g_\zeta(k_x,\kappa) = \begin{vmatrix}
		k_x  u_\zeta(k_x,\kappa) + a \kappa v_\zeta (k_x,\kappa) & k_x u_\infty(k_x,\kappa_\mathrm{ev}) + a \kappa_\mathrm{ev}  v_\infty (k_x,\kappa_\mathrm{ev})\\
		v_\zeta(k_x,\kappa) & v_\infty(k_x,\kappa_\mathrm{ev})
	\end{vmatrix},
\end{equation}
and
\begin{equation}
	h_\zeta(k_x,\kappa) = \begin{vmatrix}
		k_x  u_\zeta(k_x,\kappa) + a \kappa v_\zeta (k_x,\kappa) & k_x  u_\zeta(k_x,-\kappa) - a \kappa v_\zeta (k_x,-\kappa)\\
		v_\zeta(k_x,\kappa) & v_\zeta(k_x,-\kappa)
	\end{vmatrix}.
\end{equation}
For $(k_x,\kappa) \in S^2 \setminus\{\zeta\}$, $\widehat \psi^\zeta$ and $\widehat \psi^\infty$ are related through $t_\infty^\zeta$, leading to \eqref{eq:transition_S} by inspection of \eqref{eq:def_g_zeta}.
\end{proof}

\begin{proof}[Proof of Thm.~\ref{thm:main}, eq.~\eqref{eq:BSC}.]
We compute the Chern number as the winding number of the transition function between two sections that cover $S^2$. The key point is that the scattering state naturally provides two such sections, for which the transition function is the scattering amplitude. 

Let $\epsilon >0$ and $\zeta = \ii \zeta_y$ with $\zeta_y>\epsilon$. The section $\widehat \psi^\zeta(k_x,\kappa)$ is regular everywhere except at $k_x+ \ii \kappa = \zeta$. In particular, it is regular  along $\mathcal C_\epsilon$, including at $\infty$. Then consider $S_\zeta(k_x,\kappa) \widehat{\psi}^\zeta(k_x,\kappa)$. This section is regular for all $\kappa >0$. Indeed,
\begin{align}
	S_\zeta(k_x,\kappa) \widehat{\psi}^\zeta(k_x,\kappa) & =  \dfrac{t_\infty^\zeta(k_x,-\kappa)}{t_\infty^\zeta(k_x,\kappa)}\dfrac{t_\zeta^0(k_x,\kappa)}{t_\infty^0(k_x,-\kappa)}S_0(k_x,\kappa)\dfrac{t_\infty^\zeta(k_x,\kappa)}{t_\infty^0(k_x,\kappa)}\widehat \psi^{0}(k_x,\kappa)\cr
	&=S_0(k_x,\kappa) \dfrac{t_\infty^\zeta(k_x,-\kappa)}{t_\infty^0(k_x,-\kappa)} \widehat \psi^{0}(k_x,\kappa)\;,
\end{align}
where we have used \eqref{eq:section_zeta} and \eqref{eq:transition_S}. Notice that $|S_\zeta(k_x,\kappa)|=1$ for $k_x + \ii \kappa \in S^2 \setminus\{\zeta\}$ thanks to Prop.~\ref{prop:S_U1}. In particular, $S_0$ is regular for $\kappa >0$. Finally, for $k_x,\kappa \in \mathcal C_\epsilon$, the transition between the two sections is by definition $S_\zeta$ so that
\begin{equation}
	C_+ = \dfrac{1}{2\pi \ii}\int_{\mathcal C_{\epsilon}} S^{-1}_\zeta \dd S_\zeta
\end{equation}
which concludes the proof. 
\end{proof}

\subsection{\texorpdfstring{Relative Levinson's theorem for finite $k_x$}{Relative Levinson's theorem for finite kx}\label{sec:relative_Levinson}} 
The proof of Thm.~\ref{thm:main}, eq.~\eqref{eq:Relative_Levinson} is a direct consequence of
\begin{thm}[{\cite[Thm.~6.11]{GrafPorta13}}] \label{thm:GP} Let $\epsilon >0$ and $\zeta = \ii \zeta_y$ with $\zeta_y>\epsilon$. Let $k_x^1 < k_x^2$ that do not correspond to a crossing of an edge mode branch with the bulk region in the spectrum of \eqref{eq:ShallowWater_edge}.
Then
\begin{equation}\label{eq:thmGP}
	\lim\limits_{\epsilon \to 0} \arg S_\zeta\big((k_x,\epsilon) \big)\big|_{k_x^1}^{k_x^2} = 2\pi n(k_x^1,k_x^2)
\end{equation}
where $\arg$ denotes a continuous argument and $n(k_x^1,k_x^2)$ is the signed number of edge mode branches emerging $(+)$ or disappearing $(-)$ at the lower band limit between $k_{x,1}$ and $k_{x,2}$, as $k_x$ increases.
\end{thm}
In particular for $k_{x,1}<0$ and $k_{x,2}>0$ large enough one has $n(k_x^1,k_x^2) = n_\mathrm{b}$. Moreover, $\mathcal C_{\epsilon,\lambda_0}^\perp \simeq \{(\tfrac{1}{\lambda_x},\epsilon) | \lambda_x \in \mathbb R \setminus [-\lambda_0,\lambda_0]\}$ when $\epsilon \rightarrow 0$, so that for $\lambda_0$ small enough and with the orientation of $C_{\epsilon,\lambda_0}^\perp$, \eqref{eq:thmGP} is equivalent to \eqref{eq:Relative_Levinson}. We refer to \cite{GrafPorta13} for the proof of  Thm.~\ref{thm:GP}, which is quite general and applies to our continuous model because \eqref{eq:thmGP} is valid as long as $(k_x,\kappa)$ belong to a finite path that does not cross $\infty$. Below we briefly illustrate the main elements of the proof in our explicit model, see also Fig.~\ref{fig:exmp_levinson}, in order to compare with the anomalous case at $\infty$ in the next section. Furthermore, notice that this statement has been also checked numerically for the shallow-water model with boundary condition \eqref{eq:boundary_condition} and $a= \pm 1$ in \cite{TauberDelplaceVenaille19bis}.

\begin{figure}[htb]
	\includegraphics[scale=1]{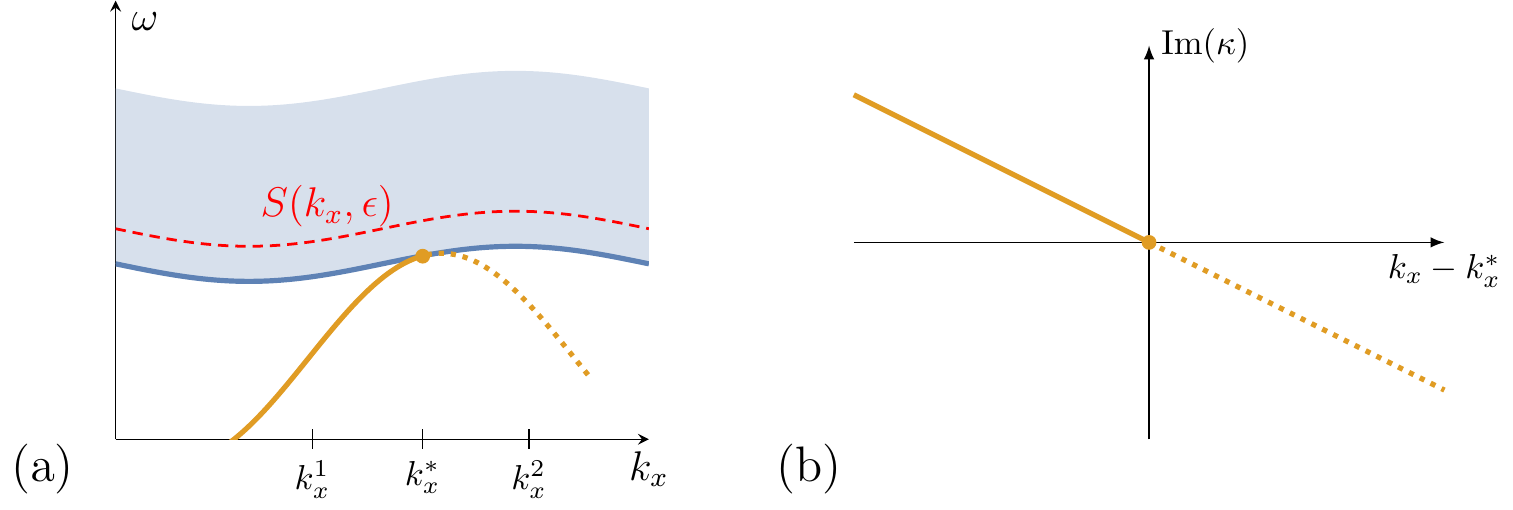}
	\caption{\label{fig:exmp_levinson}Relative Levinson's theorem (a) Near the bottom of the bulk continuum (dashed red and shaded blue), $\arg(S)$ changes by $2\pi$ around the merging point $k_x^*$ of an edge mode branch (plain yellow). The latter becomes a branch of divergent states that is not part of the edge spectrum (dotted yellow) (b). Locally, such branches can be inferred from the poles of $S$ up to analytic continuation in $\kappa$. These poles correspond to bound states for $\mathrm{Im}(\kappa) >0$ only.}
\end{figure}

Let $k_x \in \mathbb R$ be fixed and remove it for a while. We also drop the singularity $\zeta$ and assume that the sections are regular in the region of interest. Up to a multiplication of \eqref{eq:scat_explicit} by $g_\zeta \equiv g$, an equivalent scattering state is
\begin{equation}\label{eq:alternative_psi_scat}
	\phi_s = g(\kappa) \widehat \psi(-\kappa) \ee^{-\ii \kappa y} - g(-\kappa) \widehat \psi(\kappa) \ee^{\ii \kappa y} - h(\kappa) \widehat \psi(\kappa_\mathrm{ev}) \ee^{\ii \kappa_\mathrm{ev} y},
\end{equation}
see \eqref{eq:S_vs_g} and \eqref{eq:def_g_zeta}. So far we focused on $\kappa >0$, but it turns out that a neighborhood of the bulk region in the edge spectrum, and in particular edge mode branches, can be studied through the complex continuation of $\kappa$ in the scattering state. Indeed, assume that $g(\kappa) = 0$ for $\mathrm{Im}(\kappa)>0$. The first term in \eqref{eq:alternative_psi_scat} vanishes whereas the second is exponentially decaying in $y$, similarly to the third one which is evanescent. In that case $\phi_s$ is a bound state of the edge spectrum and corresponds to a point of the edge mode branch below the bulk region. 

Then, as $k_x$ varies, the zero of $g$ might move from $\mathrm{Im}(\kappa)>0$ to $\mathrm{Im}(\kappa)<0$. In the latter case the second term in \eqref{eq:alternative_psi_scat} would be exponentially diverging in $y$ and $\phi_s$ would not be a bound state anymore. Thus at $k_x$ where $\mathrm{Im}(\kappa) =0$ the edge mode branch merges with the bulk continuum. Furthermore, a zero of $g$ is a pole for $S$, and such a sign change in $\mathrm{Im}(\kappa)$ induces a $2\pi$ shift in the argument of $S$. This relative version of Levinson's theorem is proved in a general framework in \cite{GrafPorta13} via the argument's principle. Here we illustrate it on a canonical form of $S$: assume that $g = k_x - \ii \kappa$, so that $g=0$ for $\kappa_0 = -\ii k_x$. For $k_x<0$ one has a bound state, which vanishes for $k_x>0$. This means that an edge mode branch has merged with the continuum at $k_x^*=0$. The $S$ matrix reads
\begin{equation}
	S(k_x,\kappa) = -\dfrac{k_x + \ii \kappa}{k_x - \ii \kappa}
\end{equation}
Thus, for $\kappa = \epsilon >0$,  $\arg S(k_x,\epsilon)$ is shifted by $-2\pi$ as $k_x$ goes from $-\infty$ to $\infty$. Equivalently, the argument of $S$ between some $k_x<0$ and $k_x>0$ finite is $-2\pi$ when $\epsilon \rightarrow 0$. 

\subsection{Failure of Levinson's theorem at infinity \label{sec:Levinson_violation}}
Let $\lambda_0$ be small enough such that \eqref{eq:Relative_Levinson} is true. In that case the only possible singularity for the scattering amplitude along $\mathcal C_{\epsilon,\lambda_0}$ is at infinity. So far, we have worked with $S_\zeta$ with $\zeta = \ii \zeta_y$ and $\zeta_y>0$. In order to study the neighborhood of $\infty$, it is rather convenient to work with $\zeta=0$ instead. This has no influence on \eqref{eq:Levinson_violation}. Indeed, the two are related by
\begin{equation}
	S_\zeta(k_x,\kappa) = \dfrac{t^\zeta_\infty(k_x,-\kappa)}{t^\zeta_\infty(k_x,\kappa)} \dfrac{t^0_\infty(k_x,\kappa)}{t^0_\infty(k_x,-\kappa)} S_0(k_x,\kappa),
\end{equation}
and such a transition function is regular at $\infty$, so that
\begin{equation}
	\lim\limits_{\epsilon \to 0} \dfrac{1}{2\pi \ii}  \int_{\mathcal C_{\epsilon,\lambda_0}} S_\zeta^{-1} \dd S_\zeta = \lim\limits_{\epsilon \to 0} \dfrac{1}{2\pi \ii}  \int_{\mathcal C_{\epsilon,\lambda_0}} S_0^{-1} \dd S_0.
\end{equation}
The scattering state \eqref{eq:scat_explicit} reads $\widetilde \psi_\mathrm{s}^0 = \psi_\mathrm{in}^0 + S_0 \psi_\mathrm{out}^0 + T_0 \psi_\mathrm{ev}^\infty$. It involves
\begin{equation}\label{eq:section_0}
	\widehat \psi^0(\bvec{k}) = \dfrac{1}{\sqrt{2}}\dfrac{1}{k_x+ \ii k_y} \begin{pmatrix}
		\bvec{k}^2/\omega_+ \\ k_x-\ii k_y q \\ k_y+ \ii k_x q 
	\end{pmatrix},
\end{equation}
which appears to be dual to $\widehat \psi^\infty$, see \eqref{eq:section_infty}. According to the previous section, the existence of edge modes near $\infty$ is encoded in the poles of $S_0(k_x,\kappa)$ for $\mathrm{Im}(\kappa)>0$, or equivalently in the zeros of $g_0$. Using \eqref{eq:def_g_zeta}, \eqref{eq:section_infty} and \eqref{eq:section_0} we compute $g_0$ to leading order near $(k_x,\kappa) \rightarrow  \infty$:
\begin{equation}\label{eq:g0_asym}
	g_0(k_x,\kappa) \sim \ii (2 k_x + \ii a (\kappa_\mathrm{ev}(k_x,\kappa)- \kappa)),
\end{equation}
with $\kappa_\mathrm{ev}(k_x,\kappa) \sim  \ii \sqrt{2 k_x^2 + \kappa^2}$, see \eqref{eq:def_kappa_ev}. 

\paragraph{Winding number at infinity.} In terms of the dual variables \eqref{eq:dual_variables} that parametrize $\mathcal C_{\epsilon,\lambda_0}$ one has for $\lambda_x$ near $0$
\begin{equation}
	g_0(k_x,\kappa) \sim \dfrac{\ii}{\lambda_x^2 + \epsilon^2} (2 \lambda_x+ a \sqrt{2\lambda_x^2 + \epsilon^2} - \ii a \epsilon)
\end{equation}
By \eqref{eq:S_vs_g}, the winding number of $S_0$ is, up to a sign, twice the one of $g_0$. Such a winding when $\epsilon \to 0$ and $\lambda_x$ finite can be inferred by the winding  from $\lambda_x = -\infty$ to $\lambda_x =+\infty$ with $\epsilon$ finite. The prefactor on the right hand side can be ignored. As for the real part of the rest,
\begin{equation}
	2 \lambda_x+ a \sqrt{2\lambda_x^2 + \epsilon^2} \to \left\lbrace \begin{array}{ll}
		\sqrt{2}(\sqrt{2}+a)\cdot(+\infty),& \lambda_x \to + \infty \\
		\sqrt{2}(\sqrt{2}-a)\cdot(-\infty),& \lambda_x \to - \infty.
	\end{array}\right.
\end{equation}
Thus, for $|a|<\sqrt{2}$, this rest covers the whole real line as $\lambda_x$ does. We recall that $\mathcal C_{\epsilon}$ is parametrized with reverse orientation for $\lambda_x$. Hence $g_0$ winds along $\mathcal C_{\epsilon,\lambda_0}$ going in counter-clockwise direction by $-\mathrm{sign}(a) \pi$, and $S_0$ winds by $\mathrm{sign}(a)2\pi$. For $|a|>\sqrt{2}$, it does not cover the whole line so $g_0$ and $S_0$ do not wind. This proves \eqref{eq:Levinson_violation}. In particular we deduce that the transitions occur at $a = \pm \sqrt{2},\,0$. Notice that, together with \eqref{eq:BSC} and \eqref{eq:Relative_Levinson}, we infer the value of $n_\mathrm{b}$ claimed in Prop.~\ref{prop:nb_finite_kx}.

\paragraph{Analytic continuation.} According to Thm. \ref{sec:relative_Levinson}, the existence of edge modes near the bulk continuum is related to the zeros of $g_0(k_x, \kappa)$ for $\mathrm{Im}(\kappa) >0$. By inspection of \eqref{eq:g0_asym}, it seems that $\kappa = \ii c k_x$ for some $c \in \mathbb R$ could lead to zeros of $g_0$ in its asymptotic form near $\infty$. This expression involves a square-root through $\kappa_\mathrm{ev}$. It is not clear, though, which branch should be taken when $\kappa$ becomes complex. Moreover a naive computation leads to $\kappa_\mathrm{ev} \sim - \ii \tilde c |k_x|$ when $\kappa = \ii c k_x$, which suggests that $\kappa$ and $\kappa_\mathrm{ev}$ play a dual role near infinity. Thus, in order to take into account all the possible edge modes, we look for zeros of 
\begin{equation}\label{eq:defG0}
	G_0(k_x,k_{y+},k_{y-}) = \ii (2 k_x + \ii a (k_{y-} - k_{y+}))
\end{equation}
with
$
k_x^2 + k_{y\pm}^2 = X_\pm
$,
where $X_\pm$ are the solutions of $\omega^2 = X + (f-\nu X)^2$ for a given $\omega>0$. In particular, for $k_{y+} = -\sqrt{X_+-k_x^2}:=\kappa$ then $k_{y-} = -  \ii \sqrt{k_x^2-X_-} = \kappa_\mathrm{ev}(k_x,\kappa)$ so that $G_0$ and $g_0$ coincide. However the definition of $G_0$ avoids specifying any branch for the square root. Moreover, at leading order in $\omega \rightarrow \infty$, one has 
\begin{equation}\label{eq:def_kypm}
	k_x^2 + k_{y\pm}^2 = \pm \dfrac{|\omega|}{\nu}
\end{equation}
which implies
\begin{equation}\label{eq:cond_kypm}
2k_x^2 + k_{y+}^2 + k_{y-}^2 = 0 \,.
\end{equation}
Together with \eqref{eq:defG0}, the zeros of $G_0$ prompts the ansatz
\begin{equation}\label{eq:ansatz}
	k_{y\pm} = \pm \ii c_{\pm} k_x
\end{equation}
for $c_\pm \in \mathbb R$. The different choice of sign for $k_{y\pm}$ is conventional but allows for a symmetry between $c_+$ and $c_-$. Thus, a zero of $G_0$ implies
\begin{subequations}\label{eq:sol_c_pm}
	\begin{align}
		&2  + a (c_+ + c_-) = 0\,, \\
		&c_+^2 + c_-^2 = 2\,. 
	\end{align}
\end{subequations}
The solution of this system is represented diagrammatically in Fig.~\ref{fig:diagram} as the intersection between a circle and a straight line that depends on $a$. Moreover, due to \eqref{eq:def_kypm}, 
\begin{equation}
k_x^2 (1 - c_{\pm}^2) = \pm \dfrac{|\omega|}{\nu},
\end{equation}
so that $c_+^2<1$ and $c_-^2>1$. Thus with the ansatz \eqref{eq:ansatz} there is a unique pair $(c_+,c_-)$ corresponding to zeros of $G_0$, and hence of $g_0$. Finally, such zeros are associated to a bound state only if $\mathrm{Im}(k_{y\pm})>0$. This requires $c_+>0$, $c_-<0$ for $k_x \to  \infty$ and $c_+<0$, $c_->0$ for $k_x\to -\infty$. Consequently, according to Fig.~\ref{fig:diagram}, there exists an edge mode in the neighborhood of $\infty$ for $a > \sqrt{2}$ and $k_x \to \infty$, or $a < -\sqrt{2}$ and $k_x \to -\infty$ and there is no edge mode otherwise, as claimed below eq.~ \eqref{eq:Levinson_violation} in Thm.~\ref{thm:main}.
 
 \begin{figure}[htb]\centering
 	\includegraphics[scale=1]{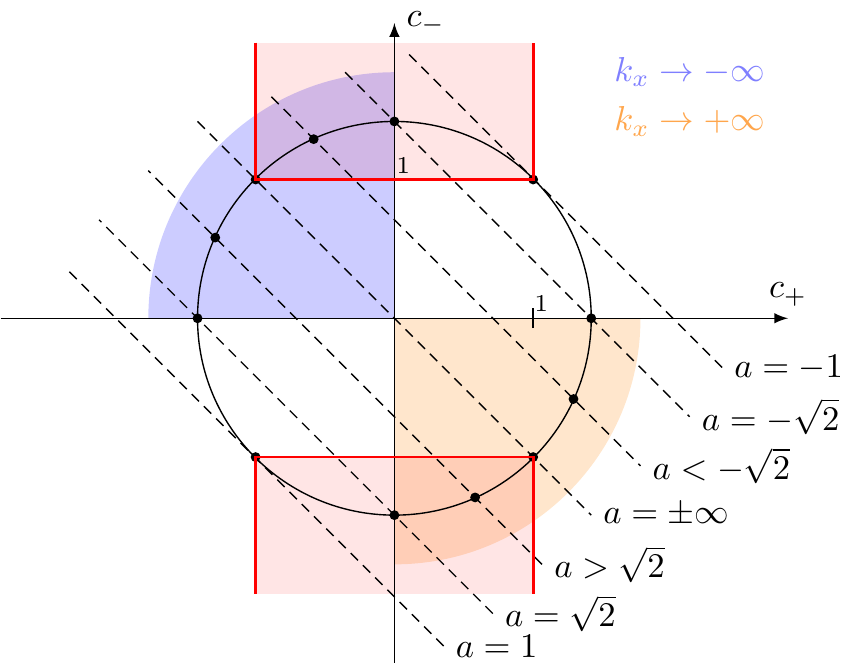}
 	\caption{\label{fig:diagram} $(c_+,c_-)$-diagram corresponding to a zero of $G_0$ near $\infty$ through the ansatz \eqref{eq:ansatz}. The solution is at the intersection between the circle and a straight line of slope $-1$ and depending on $a$. It is actually also restricted inside the red areas, leading to a single pair $(c_+,c_-)$ for each $a$.
 	Such a solution corresponds to a bound state only in the upper left or lower right quadrant, depending on the sign of $k_x$. } 
 \end{figure}

\paragraph{Two alternatives to Levinson scenario.}  Beyond the proof of Thm.~\ref{thm:main} we provide an interpretation of the mismatch when compared with Levinson 's scenario described in Sect.~\ref{sec:relative_Levinson}. We focus on what happens near the transition $a=\sqrt{2}$ for clarity. A zero of $G_0$ (and hence of $g_0$) with the ansatz \eqref{eq:ansatz} corresponds to a bound state only if both $\mathrm{Im}(k_{y+})>0$ and $\mathrm{Im}(k_{y-})>0$. In the usual Levinson's scenario of Thm.~\ref{thm:GP}, one always has $k_{y-} = \ii |\kappa_\mathrm{ev}|$, so that the nature of the state only depends on $\mathrm{Im}(k_{y+}) = \mathrm{Im}(\kappa)$, as illustrated in Fig.~\ref{fig:exmp_levinson}(b). Near $\infty$ instead, this is not the case, due to the particular structure of $g_0$ there. Two alternatives occur:
\begin{enumerate}
	\item For $a<\sqrt{2}$, one has $-1<c_+<0$ and $c_-<-1$ so that $\mathrm{Im}(k_{y+})>0$ and $\mathrm{Im}(k_{y-})<0$ for $k_x<0$, and conversely for $k_x>0$. Thus, the scattering state is nowhere bounded because it always contains some divergent part. Yet, the winding of $S_0$ is $2\pi$. 
	\item For $a>\sqrt{2}$, one has $0<c_+<1$ and $c_-<-1$ so that $\mathrm{Im}(k_{y+})<0$ and $\mathrm{Im}(k_{y-})<0$ for $k_x<0$, and conversely for $k_x>0$. Thus, the scattering state is divergent for $k_x<0$ and a bound state for $k_x>0$. A bound state emerges at $\infty$, and yet the winding of $S_0$ is $0$. 
\end{enumerate}
The two scenarios are illustrated in Fig.~\ref{fig:alternative_Levinson}. We use the dual variables \eqref{eq:dual_variables} so that $k_x\to\pm\infty$ is replaced by $\lambda_x \to 0^\pm$ which makes the comparison with Fig.~\ref{fig:exmp_levinson}(b) easier. Notice that the ansatz \eqref{eq:ansatz} becomes $\lambda_{y\pm} = \mp \ii c_{\pm} \lambda_x$.

 \begin{figure}[htb]\centering
	\includegraphics[scale=1]{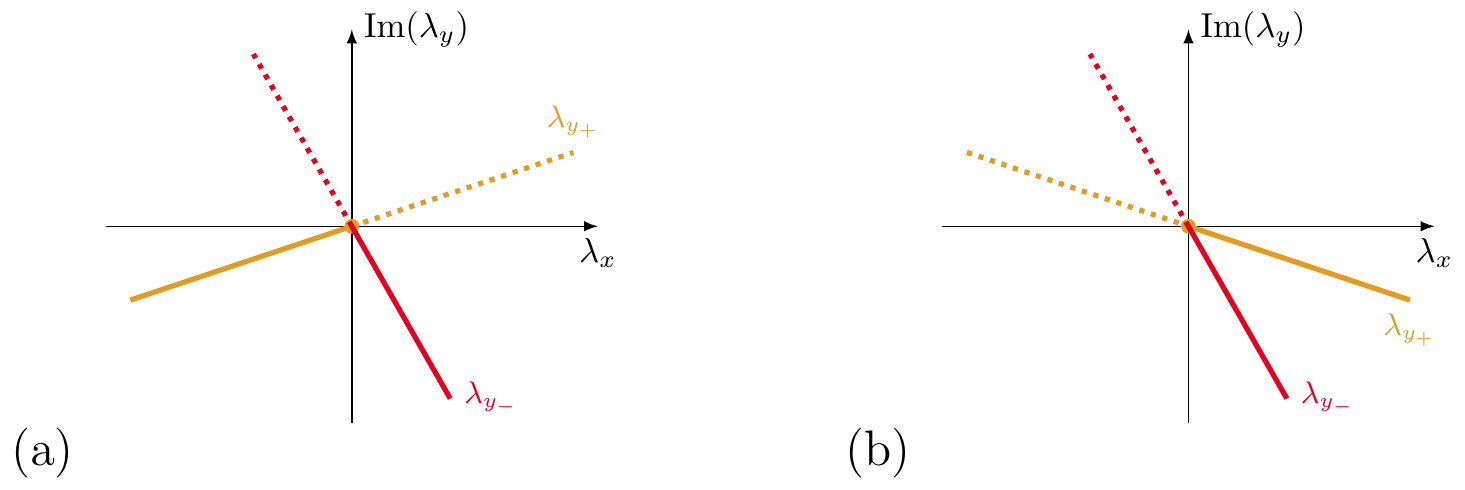}
	\caption{\label{fig:alternative_Levinson} Two alternatives to Levinson scenario at $k_x \to \pm \infty$: zeros of $G_0$ with ansatz \eqref{eq:ansatz} in terms of the dual variables. A plain (resp. dotted) line corresponds to an evanescent (resp. divergent) mode. A bound state is a superposition of two evanescent modes. (a) In case 1, there is no bound state whereas $S_0$ winds by $2\pi$. (b) In case 2, no bound state exist for $\lambda_x <0$ but one emerges at $\lambda_x=0$, whereas $S_0$ does not wind. In both cases, Levinson's theorem is violated.} 
\end{figure}

\paragraph{Second order computation.} Finally, we provide more details about the edge mode branch that exists near $|k_x|\to \infty$ for $|a|>\sqrt{2}$. As we shall see, this branch actually emerges from the bulk continuum at finite $k_x>0$ ($a>\sqrt{2}$) or $k_x<0$ ($a<\sqrt{2}$), stays close to it when $k_x \to + \infty$ or $k_x\to-\infty$, respectively,  and disappears there (cf. Fig.~\ref{fig:edge_spectrum}).  

To compute the second order correction to the result obtained we use again \eqref{eq:def_g_zeta}, \eqref{eq:section_infty} and \eqref{eq:section_0} and find
\begin{equation}\label{eq:defG1}
	G=G_0+G_1+\cdots\,,
\end{equation}
with $G_0$ as in~\eqref{eq:defG0} and 
\begin{equation}
	G_1=\frac{1}{2\nu\omega} \frac{\ii a k_x}{(k_x-\ii k_{y_-})(k_x+\ii k_{y_+})}(k_{y_+}^2-k_{y_-}^2)\,.
\end{equation}
We extend the ansatz~\eqref{eq:ansatz} by a term of the appropriate order in $k_x$
\begin{equation}\label{eq:ansatz_second_order}
	k_{y_\pm}=\pm\ii(c_\pm k_x+d_\pm k_x^{-1})\,,
\end{equation}
and obtain
\begin{equation}
	G_0(k_x,k_{y_+},k_{y_-}) = G_{00}+\ii a(d_++d_-)k_x^{-1}
\end{equation}
where $G_{00}$ is the same expression as in~\eqref{eq:defG0} with $k_{y_\pm}$ at leading order, and hence $G_{00}=0$ for $c_\pm$ zeros of $G_0$ at first order (cf.~\eqref{eq:sol_c_pm}). Moreover, 
\begin{equation}
	G_1=\frac{1}{2\nu\omega} \frac{\ii a k_x}{(1-c_-)(1-c_+)}(c_-^2-c_+^2)\,.
\end{equation}
As here the leading order in $\omega$ suffices (cf.~\eqref{eq:def_kypm}) we find for the solutions $G=0$
\begin{equation}\label{eq:d_pm_cond_1}
	d_++d_-=-\frac{1}{\nu^2 (1-c_-)(1-c_+)}\,.
\end{equation}
Furthermore, we have to amend~\eqref{eq:cond_kypm} to
\begin{equation}
	2k_x^2+k_{y_+}^2+k_{y_-}^2=-\frac{1-2\nu f}{\nu^2}\,,
\end{equation}
which can be seen by including the term of next order in~\eqref{eq:def_kypm}. Plugging in the ansatz together with~\eqref{eq:sol_c_pm} we find that $d_\pm$ are determined by~\eqref{eq:d_pm_cond_1} and
\begin{equation}\label{eq:d_pm_cond_2}
	c_+d_++c_-d_- = \frac{1-2\nu f}{2\nu^2}\,.
\end{equation}
With these results we will focus on the transition at $a=\sqrt{2}$ and explain the emergence of an edge mode branch for $a>\sqrt{2}$. The case of negative values of $a$ works analogously. From Fig.~\ref{fig:diagram} and its explanation we see that the first order solution for $a=\sqrt{2}$ is given by $(c_+,c_-)=(0,-\sqrt{2})$ and plugging this into (\ref{eq:d_pm_cond_1},~\ref{eq:d_pm_cond_2}) we find
\begin{gather}
	d_+ +d_- =\nu^{-2}(1-\sqrt{2})\,,\\
	d_-=-\nu^{-2}\frac{1-2\nu f}{2\sqrt{2}}\,,
\end{gather}
and hence 
\begin{equation}
	d_+=\frac{1}{2\sqrt{2}\nu^2}(2\sqrt{2}-3-2\nu f)\approx \frac{2\sqrt{2}-3}{2\sqrt{2}\nu^2}\,,
\end{equation}
neglecting terms of order $1/\nu$ in the last step. A further incipient state (besides of $k_x=-\infty$) occurs when $\operatorname{Im} k_{y,\sigma}=0$ for $\sigma=+$ or $\sigma=-$. Thus by the ansatz~\eqref{eq:ansatz_second_order} the imaginary part of $k_{y_\pm}$ changes sign at the values
\begin{equation}\label{eq:second_order_sol}
	k_x^2=-\frac{d_\pm}{c_\pm}\,.
\end{equation}
Since $c_-\neq0$ at $a=\sqrt{2}$, the solution in the $-$ case is likely outside the range of validity of the expansion. But for $a\searrow\sqrt{2}$ we have $c_+\searrow 0$ (see Fig.~\ref{fig:diagram}). Thus and by $d_+<0$ there is a solution $k_x$ of~\eqref{eq:second_order_sol} with $k_x\rightarrow \infty$ in agreement with Fig.~\ref{fig:edge_spectrum}.

\appendix

\section{\texorpdfstring{Chern number for spin $s$ representations}%
                               {Chern number for spin s representations}\label{app:chern_spin}}

It suffices to prove Proposition \ref{prop:chern_numbers} in the case where $M=S^2$ and we do so by induction in $s$ in steps of $1/2$, starting with $s=0$ and $s=1/2$. In the first case the bundle is trivial, $S^2\times \mathbb{C}^2$, whence $C(P_{0,0})=0$, where the eigenprojection on the band with labels $(s,m)$ is denoted by $P_{s,m}$. In the case $s=1/2$, the integrand of \eqref{eq:chern_tr} is
\begin{equation}
	\tr(P_\pm[dP_\pm ,dP_\pm]) = \pm \frac{\ii}{2}\vec{e}\cdot (d\vec{e}\wedge d\vec{e}\,)=\pm \frac{\ii}{2}w
\end{equation}
with $P_\pm = P_{\frac{1}{2},\pm\frac{1}{2}}$. This result follows from
\begin{gather}
	P_\pm = \frac{1}{2}(1\pm\vec{e}\cdot\vec{\sigma})\,,\\
	\tr\,\vec{a}\cdot\vec{\sigma}\,\big[\vec{b}\cdot\vec{\sigma},\vec{c}\cdot\vec{\sigma}\big] = 4\ii\vec{a}\cdot(\vec{b}\wedge\vec{c}\,)\,,
\end{gather}
where $\vec{\sigma}=(\sigma_1,\,\sigma_2,\,\sigma_3)^T$ denotes the vector of Pauli matrices $\sigma_i$. This leads in turn to $C(P_\pm) = \pm 1$ by $\int_{S^2}w = 4\pi$.

We next assume the claim to be true up to $s$ and prove it for $s+1/2$. Let $\mathcal{D}_s$ be the irreducible representation of $SU(2)$ of spin $s$, equipped with the standard basis  ${\ket{s,m}}_{m=-s}^s$ with respect to the quantization axis $\vec{e}$, i.e., $\vec{S}\cdot\vec{e}\,\ket{s,m}=m\ket{s,m}$. We then have by the Clebsch-Gordan series
\begin{equation}\label{eq:repr_sum}
	\mathcal{D}_{s+\frac{1}{2}}\oplus\mathcal{D}_{s-\frac{1}{2}} = \mathcal{D}_s\otimes\mathcal{D}_{\frac{1}{2}}\,.
\end{equation}
We first treat the case $m=s+1/2$, for which we find
\begin{equation}
	\ket{s+\frac{1}{2},s+\frac{1}{2}}=\ket{s,s>\otimes\,|\frac{1}{2},\frac{1}{2}},
\end{equation}
whence
\begin{equation}
	P_{s+\frac{1}{2},s+\frac{1}{2}}=P_{s,s}\otimes P_{\frac{1}{2},\frac{1}{2}}\,.
\end{equation}
Since the vector bundles are line bundles, the Chern number is additive, meaning
\begin{equation}
	C(P_{s+\frac{1}{2},s+\frac{1}{2}})=C(P_{s,s})+C(P_{\frac{1}{2},\frac{1}{2}})=2s+1 \,,
\end{equation}
as claimed. The case $m=-(s+1/2)$ is similar. Finally, we consider the intermediate cases $m=-(s-1/2),...,s-1/2$. The eigenspace of (total) $\vec{S}\cdot\vec{e}$ acting on \eqref{eq:repr_sum} for eigenvalue $m$ has dimension $2$ and can be represented as  a span in two ways:
\begin{equation}
	\Bigg[\ket{s,m-\frac{1}{2}}\otimes\,\ket{\frac{1}{2},\frac{1}{2}},\, \ket{s,m+\frac{1}{2}}\otimes\,\ket{\frac{1}{2},-\frac{1}{2}}\Bigg] = \Bigg[\ket{s-\frac{1}{2},m},\,\ket{s+\frac{1}{2},m}\Bigg]\,.
\end{equation}
The bundle over $S^2\ni\vec{e}$ having the eigenspaces as fibers is thus
\begin{equation}
	\Big(P_{s,m-\frac{1}{2}}\otimes P_{\frac{1}{2},\frac{1}{2}}\Big)\oplus\Big(P_{s,m+\frac{1}{2}}\otimes P_{\frac{1}{2},-\frac{1}{2}}\Big)=P_{s-\frac{1}{2},m}\oplus P_{s+\frac{1}{2},m}\,.
\end{equation}
Arguing as before we get for $c_{s,m}:=C(P_{s,m})$
\begin{equation}
	(c_{s,m-\frac{1}{2}}+c_{\frac{1}{2},\frac{1}{2}})+(c_{s,m+\frac{1}{2}}+c_{\frac{1}{2},-\frac{1}{2}})=c_{s-\frac{1}{2},m}+c_{s+\frac{1}{2},m}\,,
\end{equation}
i.e.
\begin{equation}
	((2m-1)+1)+(2m+1-1) = 2m + c_{s+\frac{1}{2},m}
\end{equation}
by induction assumption. Thus $c_{s+\frac{1}{2},m}=2m$ which proves Proposition~\ref{prop:chern_numbers}.

\section{Self-adjoint boundary conditions\label{app:selfadjoint_bc}}

In this appendix we will characterize self-adjoint boundary conditions for our model with domain $\left\{(x,y)\,|\, y \geq 0\right\}\subset\mathbb{R}^2$. The Hamiltonian is after Fourier transformation along $x$ with conjugate variable $k_x$ (translation invariance in $x$-direction) given by (cf. \eqref{eq:ShallowWater_edge})
\begin{equation}
	H^\sharp(k_x)=\begin{pmatrix}
		0 & k_x & -\ii\partial_y\\
		k_x & 0 & -\ii (f-\nu (k_x^2-\partial_y^2))\\
		-\ii\partial_y & \ii (f-\nu (k_x^2-\partial_y^2)) & 0
	\end{pmatrix}\,.
\end{equation}
In this appendix we will drop the $\tilde{\cdot}$ as compared to the main text and denote the states by
\begin{equation}\label{eq:psi_vector}
	\psi = \psi(k_x;y) =\begin{pmatrix}
	\eta\\
	u\\
	v
	\end{pmatrix}\,.
\end{equation}
\begin{lem}\label{lem:self_adj_subspaces}
	Self-adjoint realizations of the Hamiltonian $H$ correspond to subspaces $M\subset \mathbb{C}^6$ with 
	\begin{equation}\label{eq:cond_omega}
		\Omega M = M^\perp \,,
	\end{equation}
	where 
	\begin{equation}\label{eq:Omega}
		\Omega=\begin{pmatrix}
			0 & 0 & -1 & 0 & 0 & 0\\
			0 & 0 & 0 & 0 & 0 & -\nu\\
			-1 & 0 & 0 & 0 & \nu & 0\\
			0 & 0 & 0 & 0 & 0 & 0\\
			0 & 0 & \nu & 0 & 0 & 0\\
			0 & -\nu & 0 & 0 & 0 & 0
		\end{pmatrix}\,.
	\end{equation}
	The correspondence is in terms of $\psi\in H\oplus H^2\oplus H^2$:
	\begin{equation}
		\psi \in \mathcal{D}(H)\longleftrightarrow \Psi\in M\,,
	\end{equation}
	where the (stacked) column vector
	\begin{equation}\label{eq:stacked_v}
		\Psi=\begin{pmatrix}
			\psi\\
			\psi'
		\end{pmatrix}
	\end{equation}
	stands for the boundary values at $y=0$.
\end{lem}
\begin{proof}
	Without yet imposing any boundary condition, the fact that $\psi\in H\oplus H^2\oplus H^2$ implies that $\eta$, $u'$, $v'$ are continuous and vanish at infinity. A partial integration of $\braket{\tilde{\psi}}{H^\sharp\psi}$ thus yields boundary terms at $y=0$ only:
	\begin{align}
		-\ii \left(\braket{\tilde{\psi}}{H^\sharp\psi}-\braket{H^\sharp\tilde{\psi}}{\psi}\right) &= -(\bar{\tilde{\eta}}v + \bar{\tilde{v}}\eta) + \nu(\bar{\tilde{v}}u'-\bar{\tilde{v}}'u)-\nu(\bar{\tilde{u}}v'-\bar{\tilde{u}}'v)\nonumber\\
		&=\tilde{\Psi}^*\Omega\Psi\label{eq:boundary_cond_part}
	\end{align}
	with $\Psi$ as in \eqref{eq:stacked_v}, $\Omega$ as in the Lemma to prove and $'=\partial_y$. Let the domain be $\left\{\Psi | \Psi \in M \right\}$, where $M\subset\mathbb{C}^6$ is some subspace. Then $M$ should have the properties
	\begin{align*}
		\tilde{\Psi}^*\Omega\Psi =0\,, \: (\Psi\in M) &\Rightarrow \tilde{\Psi}\in M\,,\\
		\Psi\in M &\Rightarrow \tilde{\Psi}^*\Omega\Psi =0 \:(\forall \tilde{\Psi}\in M)\,.
	\end{align*}
	In fact the first one implies $H^*\subset H$ and the second $H\subset H^*$, whence $H=H^*$ as required. Because of $\Omega^*=\Omega$ the two properties are summarized by
	\begin{equation}
		\tilde{\Psi}^*\Omega\Psi =0\,,\:(\Psi\in M) \iff \tilde{\Psi}\in M\,,
	\end{equation}
	which is in turn equivalent to $(\Omega M)^\perp = M$, i.e. to $\Omega M = M^\perp$.
\end{proof}
We note that $\rk\Omega = 4$ and
\begin{equation}\label{eq:orthog_direct_sum}
	\ker\Omega\oplus\im\Omega = \mathbb{C}^6
\end{equation}
(orthogonal direct sum) by $\Omega=\Omega^*$. Dimensions are $2$ and $4$. Let $\hat{\Omega}$ be a partial left inverse of $\Omega$:
\begin{equation}\label{eq:partial_left_inverse}
	\hat{\Omega}\Omega = P\,,
\end{equation}
where $P$ is the orthogonal projection on $\im\Omega$ associated to \eqref{eq:orthog_direct_sum}. It follows 
\begin{equation}\label{eq:partial_left_inverse_i}
	\Omega\hat{\Omega}v= v ,\quad (v\in\im\Omega)\,.
\end{equation}

\begin{lem}\label{lem:equivalent_boundary_cond}
	Let $M\subset\mathbb{C}^6$. The following are equivalent:
	\begin{enumerate}[a)]
		\item $\Omega M = M^\perp$
		\item There is a subspace $\tilde{M}\subset \mathbb{C}^6$ such that 
			\begin{enumerate}[1)]
				\item $M=\ker\Omega\oplus\tilde{M}$, orthogonal direct sum, (whence $\tilde{M}\subset\im\Omega$),
				\item $\hat{\Omega}\tilde{M}^\perp =\tilde{M}$, where $\tilde{M}^\perp$ is the orthogonal complement of $\tilde{M}$ within $\im\Omega$.
			\end{enumerate}
		\item \begin{enumerate}[1)]
			\item $\ker\Omega\subset M$,
			\item $\dim M =4$,
			\item $\hat{\Omega}^\perp\subset M$.
		\end{enumerate}
	\end{enumerate}
\end{lem}
\begin{proof}
	\textit{(a)}$\Rightarrow$\textit{(b)}: By \textit{(a)}, $M^\perp\subset\im\Omega$, and thus by \eqref{eq:orthog_direct_sum} $M\supset\ker\Omega$, proving \textit{(b1)}. Next we have $M^\perp=\tilde{M}^\perp$ as we find by \textit{(b1)}
	\begin{align*}
		v\in M^\perp &\iff v\perp \ker\Omega ,\: v\perp \tilde{M}\\
		&\iff v\in \tilde{M}^\perp
	\end{align*}
	because $(\ker\Omega)^\perp=\im\Omega$. With \eqref{eq:partial_left_inverse} we get from \textit{(a)} $PM=\hat{\Omega}M^\perp$, i.e. $\tilde{M}=\hat{\Omega}\tilde{M}^\perp$.
	
	\textit{(b)}$\Rightarrow$\textit{(c)}: First we see directly that \textit{(c1)} follows from \textit{(b1)}. Furthermore, since $\hat{\Omega}$ is regular as a map $\im\Omega\rightarrow\im\Omega$, we have by \textit{(b2)}: $4-\dim\tilde{M}=\dim\tilde{M}$, i.e. $\dim\tilde{M}=2$, proving \textit{(c2)}. Property \textit{(c3) }follows from \textit{(b2)} and $M^\perp=\tilde{M}^\perp$. 
	
	\textit{(c)}$\Rightarrow$\textit{(a)}: By \textit{(c1)}, i.e. $\im\Omega \supset M^\perp$, and \eqref{eq:partial_left_inverse_i} we get from \textit{(c3)}
	\begin{equation}\label{eq:M_perp_sub}
		M^\perp \subset \Omega M\,.
	\end{equation}
	By the rank-nullity theorem applied to $\Omega : M\rightarrow\mathbb{C}^6$, i.e.
	\begin{equation}
		\dim M = \dim\ker (\Omega\restriction M) + \dim \Omega M\,,
	\end{equation}		
	we get $4=2+\dim\Omega M$ by \textit{(c1,c2)}. Hence equality in \eqref{eq:M_perp_sub}.
\end{proof}

\subsection{Boundary conditions in terms of equations}

In the following the self-adjoint boundary conditions will be characterized more explicitly in terms of equations. For that we observe that $\ker\Omega$ is spanned by the columns of the matrix
\begin{equation}\label{eq:matrix_span}
	N=\begin{pmatrix}
		\nu & 0\\
		0 & 0\\
		0 & 0\\
		0 & 1\\
		1 & 0\\
		0 & 0	
	\end{pmatrix}
\end{equation}
The partial left inverse $\hat{\Omega}$ is uniquely determined on $\im\Omega$, but is arbitrary on $\ker\Omega$. For definiteness, let us choose $\hat{\Omega}v=0$, $(v\in\ker\Omega)$. By that we have explicitly
\begin{equation}
	\hat{\Omega} = \begin{pmatrix}
		0 & 0 & -\lambda & 0 & 0 & 0\\
		0 & 0 & 0 & 0 & 0 & - \nu^{-1}\\
		-\lambda  & 0 & 0 & 0 & \lambda \nu & 0\\
		0 & 0 & 0 & 0 & 0 & 0\\
		0 & 0 &  \lambda \nu & 0 & 0 & 0\\
		0 & - \nu^{-1} & 0 & 0 & 0 & 0
	\end{pmatrix}\,,\quad \lambda = \frac{1}{1+\nu^2}\,.
\end{equation}
In fact the so chosen partial left-inverse fulfills $\Omega\hat{\Omega} = P$ and $\hat{\Omega}N =0$ (cf. \eqref{eq:partial_left_inverse}).
\begin{prop}\label{prop:boundary_cond}
	\begin{enumerate}[i)]
		\item Subspaces $M\subset\mathbb{C}^6$ as in Lemma~\ref{lem:self_adj_subspaces} of dimension $4$ are determined by $2\times6$ matrices $A$ of maximal rank, i.e. $\rk A=2$, by means of
		\begin{equation}\label{eq:subspaces_dim4}
			M= \{\Psi\in\mathbb{C}^6 | A\Psi =0 \} = \ker A\,,
		\end{equation}	
		and conversely. Two such matrices $A$, $\tilde{A}$ determine the same subspace if and only if $A=B\tilde{A}$ with $B\in \mathrm{GL}(2)$.	 
		\item Self-adjoint boundary conditions are determined precisely by matrices as in $(i)$ with 
		\begin{equation}\label{eq:self_adjoint_matrix_cond}
			AN=0\,,\qquad A\hat{\Omega}A^*=0
		\end{equation}
	\end{enumerate}
\end{prop}
\begin{proof}
	\begin{enumerate}[i)]
		\item Only the last sentence deserves proof, and in fact only the necessity of $A=B\tilde{A}$. For that  consider a map $B:\,\mathbb{C}^2\rightarrow\mathbb{C}^2$ which is well-defined by $Av\mapsto\tilde{A}v$, $v\in\mathbb{C}^6$ because of $\ker A=\ker \tilde{A}$.
		\item By the Lemma \ref{lem:equivalent_boundary_cond} and in particular by the equivalence between \textit{(a)} and \textit{(c)}, $M$ is as in \eqref{eq:subspaces_dim4} by \textit{(c2)}. By \textit{(c1)} and \eqref{eq:matrix_span} the first equation~\eqref{eq:self_adjoint_matrix_cond} applies. Equation~\eqref{eq:subspaces_dim4} states $M^\perp=\ran A^* = \{A^*v|v\in\mathbb{C}^2\}$. Thus by \textit{(c3)},
		\begin{equation}
			\braket{A^*v_1}{\hat{\Omega}A^*v_2}=0\,,\quad (v_1,\,v_2\in\mathbb{C}^2)\,,
		\end{equation}
		i.e. the second equation~\eqref{eq:self_adjoint_matrix_cond}.
	\end{enumerate}
\end{proof}
\begin{exmp}
	Dirichlet boundary conditions $u=0$, $v=0$ correspond to
	\begin{equation}
		A= \begin{pmatrix}
			0 & 1 & 0 & 0 & 0 & 0\\
			0 & 0 & 1 & 0 & 0 & 0
		\end{pmatrix}
	\end{equation}
	which is seen to satisfy \eqref{eq:self_adjoint_matrix_cond}.
\end{exmp}

\subsection{Local boundary conditions}

In this section we will study local boundary conditions which contain the ones studied in the main text of the paper. For that let $\psi(\bvec{x})$, $\bvec{x}=(x,y)\equiv(x_1,\,x_2)$ be as in \eqref{eq:psi_vector} and lets consider local boundary conditions at $x_2=0$ of the form
\begin{equation}
		B_0\psi +B_1\partial_1\psi +B_2\partial_2\psi = 0
\end{equation}
($\partial_i=\partial/\partial x_i$) with $l\times 3$-matrices $B_i$ ($l=2$ suffices). After using translation invariance
\begin{equation}
	\psi(\bvec{x}) = \psi(x_2)e^{\ii k_1x_1}
\end{equation}
they reduce to 
\begin{equation}
	(B_0+\ii k_1B_1)\psi + B_2\psi ' =0\,,
\end{equation}
($'=\partial /\partial x_2$), i.e. to $A\Psi=0$ as in Proposition~\ref{prop:boundary_cond} with
\begin{equation}
	\begin{gathered}
		A=(B_0+\ii k_1B_1,\, B_2) \equiv A_{02}+\ii k_1A_1 \equiv A(k)\\
		A_{02}=(B_0,\, B_2)\,,\quad A_1=(B_1,\, 0)\,.
	\end{gathered}
\end{equation}
\begin{rem}
	Quite generally, the condition $\rk A \geq 2$ is equivalent to $A\wedge A \neq 0$. So $\rk A(k_1) \geq 2$ (and hence $=2$, c.f. Proposition \ref{prop:boundary_cond}) means that at least one among
	\begin{equation*}
		A_{02}\wedge A_{02}\,,\quad A_{02}\wedge A_{1}+A_{1}\wedge A_{02}\,,\quad A_1\wedge A_1
	\end{equation*}
	does not vanish.
\end{rem}
The conditions \eqref{eq:self_adjoint_matrix_cond} now mean
\begin{gather}
	A_{02}N=0\,,\quad A_1N=0,,\\
	A_{02}\hat{\Omega}A_{02}^*=0\,,\quad A_{02}\hat{\Omega}A_{1}^*-A_{1}\hat{\Omega}A_{02}^*=0\,,\quad A_1\hat{\Omega}A_1^*=0\,.
\end{gather}
\begin{exmp}
	 The boundary conditions $v=0$, $\partial_xu+a\partial_yv=0$ correspond to 
	 \begin{equation}
	 	A_{02}=\begin{pmatrix}
	 		0 & 0 & 1 & 0 & 0 & 0\\
	 		0 & 0 & 0 & 0 & 0 & a
	 	\end{pmatrix}\,,\quad A_1=\begin{pmatrix}
	 		0 & 0 & 0 & 0 & 0 & 0\\
	 		0 & 1 & 0 & 0 & 0 & 0
	 	\end{pmatrix}\,,
	 \end{equation}
	 which are seen to fulfill the conditions stated above.
\end{exmp}

\section{Scattering theory for general boundary conditions}

The scattering states and the scattering amplitude can be defined for any self-adjoint boundary condition. In the main text we focused on the variables $(k_x,\kappa)$, so that $\omega=\omega_+(k_x,\kappa)$. Here, we work instead with $(k_x,\omega)$. Asymptotic states are given in terms of plane wave solutions
\begin{equation}\label{eq:plane_wave_solutions}
	\psi(x,t) = \hat{\psi} e^{\ii (k_xx+k_yy-\omega t)}\,.
\end{equation}
Let $k_x$ and $\omega >\sqrt{k_x^2+(f-\nu k_x^2)^2}$ be given. There are two solutions $X_\pm\equiv \bvec{k}^2$ of
\begin{equation*}
	\omega^2=X+(f-\nu X)^2\,;
\end{equation*}
they have $\pm X_\pm >0$ because of $\nu^2 X_+X_-=f^2-\omega^2<0$; and hence $4$ solutions $k_y$ of $X=k_x^2+k_y^2$, namely two real ones, incoming ($\kappa_\mathrm{in}<0)$, outgoing ($\kappa_\mathrm{out}=-\kappa_\mathrm{in}>0$); and two imaginary ones, decaying ($\kappa_\mathrm{ev},\, \ii \kappa_\mathrm{ev} <0$) and diverging ($\kappa_\mathrm{div}=-\kappa_\mathrm{ev},\, \ii \kappa_\mathrm{div}>0$). The first three are, up to multiples, cf.~\eqref{eq:section_infty},
\begin{equation}\label{eq:solutions_xpm}
	\widehat{\psi}_\mathrm{in} = \begin{pmatrix}
		(k_x^2+\kappa_\mathrm{in}^2)/\omega\\
		k_x-\ii \kappa_\mathrm{in} q_+\\
		\kappa_\mathrm{in}+\ii k_x q_+
	\end{pmatrix}\,,\quad\widehat{\psi}_\mathrm{out} = \begin{pmatrix}
		(k_x^2+\kappa_\mathrm{out}^2)/\omega\\
		k_x-\ii \kappa_\mathrm{out} q_+\\
		\kappa_\mathrm{out}+\ii k_x q_+
	\end{pmatrix}\,,\quad \widehat{\psi}_\mathrm{ev} = \begin{pmatrix}
		(k_x^2+\kappa_\mathrm{ev}^2)/\omega\\
		k_x- \ii\kappa_\mathrm{ev} q_-\\
		 \kappa_\mathrm{ev}+\ii k_x q_-
	\end{pmatrix}\,,
\end{equation}
with 
\begin{equation}\label{eq:qpm}
	q_\pm = \frac{f-\nu X_\pm}{\omega}
\end{equation}
The solutions \eqref{eq:solutions_xpm} are to be seen in relation with \eqref{eq:plane_wave_solutions}. They contribute boundary values
\begin{equation}\label{eq:boundary_values}
	\Psi_\mathrm{in}=\begin{pmatrix}
		\widehat{\psi}_\mathrm{in}\\
		\ii \kappa_\mathrm{in}\widehat{\psi}_\mathrm{in}
	\end{pmatrix}\,,\quad \Psi_\mathrm{out}=\begin{pmatrix}
		\widehat{\psi}_\mathrm{out}\\
		\ii \kappa_\mathrm{out}\widehat{\psi}_\mathrm{out}
	\end{pmatrix}\,,\quad \Psi_\mathrm{ev}=\begin{pmatrix}
		\widehat{\psi}_\mathrm{ev}\\
		\ii \kappa_\mathrm{ev} \widehat{\psi}_\mathrm{ev}
	\end{pmatrix}\,.
\end{equation}
We shall assume that there are no embedded eigenvalues. We conjecture this to be true for any self-adjoint boundary condition, and we show it for~\eqref{eq:boundary_condition}, which is of relevance for the rest of this paper. We do so at the end of this section.
\begin{lem}\label{lem:linear_indep}
	For $k_y\neq 0$ (cf. $\omega>\sqrt{k_x^2+(f-\nu k_x^2)^2}$ above) the three vectors in~\eqref{eq:boundary_values} are linearly independent. 
\end{lem}
\begin{proof}
	Inspection of $H\psi=\omega\psi$ as a differential equation in $y$, cf. \eqref{eq:ShallowWater_edge}, shows that any initial values
	\begin{equation}
		\Psi = \begin{pmatrix}
			\psi\\
			\psi '
		\end{pmatrix}=\begin{pmatrix}
			\eta\\
			\vdots\\
			v'
		\end{pmatrix}
	\end{equation}
	determine an existing and unique solution provided
	\begin{equation}
		k_x u -\ii v'=\omega\eta\,,
	\end{equation}
	which \eqref{eq:boundary_values} do. A linear combination 
	\begin{equation}
		f_\mathrm{in}\Psi_\mathrm{in} + f_\mathrm{out}\Psi_\mathrm{out}+f_\mathrm{ev}\Psi_\mathrm{ev}=0
	\end{equation}
	then implies
	\begin{equation}
		f_\mathrm{in}\hat{\psi}_\mathrm{in}e^{\ii \kappa_\mathrm{in} y}+f_\mathrm{out}\hat{\psi}_\mathrm{out}e^{\ii \kappa_\mathrm{out} y}+f_\mathrm{ev}\hat{\psi}_\mathrm{ev}e^{\ii \kappa_\mathrm{ev} y}=0\,.	
	\end{equation}		
	Linear independence of the functions $e^{\ii \kappa_\mathrm{in} y}$, $e^{\ii \kappa_\mathrm{out} y}$, $e^{\ii \kappa_\mathrm{ev} y}$ implies $f_\mathrm{in}=f_\mathrm{out}=f_\mathrm{ev}=0$, as was to be shown. 
\end{proof}
\begin{lem}\label{lem:vanishing_solution}
	Let $\psi$, $\tilde{\psi}$ be two solutions of $H\psi=\omega\psi$ that are bounded in $y\geq0$, but regardless of boundary conditions. If one of them vanishes at $y\rightarrow +\infty$, then
	\begin{equation}\label{eq:cond_infty}
		\tilde{\Psi}^*\Omega\Psi=0\,,
	\end{equation}	 
	where $\Psi$ is given by~\eqref{eq:qpm}.
\end{lem}
\begin{proof}
	The terms on the r.h.s of~\eqref{eq:boundary_cond_part} are supposed to be evaluated at $y=0$. The equation itself was obtained because the same terms would vanish for $y\to\infty$, and that is what they still do here, because in each of them one factor does while the other stays bounded. In fact, if a solution $\psi$ is bounded (or even vanishes at infinity), i.e. $\Im{k_y}\geq 0$ in~\eqref{eq:plane_wave_solutions}, then so does $\psi'$. Moreover, the l.h.s. of~\eqref{eq:boundary_cond_part} vanishes by $H\psi=\omega\psi$.
\end{proof}

Let a boundary condition $M$ be determined by a matrix $A$ as in (\ref{eq:subspaces_dim4},~\ref{eq:self_adjoint_matrix_cond}). It reads
\begin{equation}\label{eq:boundary_cond_matrix}
	A\Psi = 0\,,\quad \text{for}\quad \Psi=f_\mathrm{in}\Psi_\mathrm{in}+f_\mathrm{out}\Psi_\mathrm{out}+f_\mathrm{ev}\Psi_\mathrm{ev}\,.
\end{equation}
\begin{lem}\label{lem:appS_unique}
	There is a unique solution $(f_\mathrm{in},\,f_\mathrm{out},\,f_\mathrm{ev})$ up to multiples. 
\end{lem}

\begin{proof}
	There is at least one solution, since the span of~\eqref{eq:boundary_values}, which has dimension $3$ by Lemma~\ref{lem:linear_indep}, must intersect non-trivially $\ker A$ (of dimension $4$) in view of $3+4>6$. 
	On the other hand, there are no two linearly independent solutions. In fact, if so, there would be a solution of \eqref{eq:boundary_cond_matrix} with $f_\mathrm{in}=0$, i.e.
	\begin{equation}
		\Psi =f_\mathrm{out}\Psi_\mathrm{out}+f_\mathrm{ev}\Psi_\mathrm{ev}\,,
	\end{equation}
	which we shall rule out, unless trivial. Since $\Psi\in M$ we have $\Psi^*\Omega\Psi=0$ by~\eqref{eq:cond_omega}. Even though~\eqref{eq:cond_infty} does not allow to reach the same conclusion for $\Psi_{out}$, because $\psi_{out}$ is not vanishing at $y\to\infty$, it does for $f_{out}\Psi_{out}=\Psi-f_{ev}\Psi_{ev}$, because $\psi_{ev}$ does  and so any contribution to~\eqref{eq:cond_infty} involving it:
	\begin{equation}
		\abs{f_{out}}^2\cdot\Psi_{out}^*\Omega\Psi_{out}=0\,.
	\end{equation}
	We claim that the second factor does not vanish. In fact, a straightforward computation based on~\eqref{eq:boundary_cond_part} gives
	\begin{align}
		\Psi_{out}^*\Omega\Psi_{out} &= -\frac{2X_+}{\omega}\kappa_{out}+4\nu X_+\kappa_{out}q_+\nonumber\\
		&= -\frac{2X_+\kappa_{out}}{\omega}\left(1-2\nu(f-\nu X_+)\right)<0
	\end{align}
	because of $1-2\nu f>0$, $2\nu^2X_+>0$.
\end{proof}

\begin{prop}\label{prop:S_U1}
	For any self-adjoint boundary condition, $k_x\in \mathbb R$ and $\omega > \sqrt{k_x^2+(f-\nu k_x^2)^2}$ the scattering amplitude $S(k_x,\omega )\equiv S = f_\mathrm{out}/f_\mathrm{in}$ is well-defined and satisfies $|S|=1$.
\end{prop}
\begin{proof}
	A straightforward computation yields
	\begin{gather}
		\Psi_{out}^*\Omega\Psi_{out} = -\Psi_{in}^*\Omega\Psi_{in}\neq 0\,,\\
		\Psi_{out}^*\Omega\Psi_{in} =0 = \Psi_{in}^*\Omega\Psi_{out}\,,
	\end{gather}
	where the last equality follows by $\Omega^*=\Omega$. Using $\Psi^*\Omega\Psi=0$ for $\Psi$ as in~\eqref{eq:boundary_cond_matrix} yields $\abs{f_{out}}^2-\abs{f_{in}}^2=0$ by~\eqref{eq:cond_infty}.
\end{proof}

\begin{rem}
	The boundary condition~\eqref{eq:boundary_condition}, and in particular $v=0$, does not allow for embedded eigenvalues. In fact, in view of~\eqref{eq:solutions_xpm}, that would amount to 
	\begin{equation}
		\kappa_{ev}+\ii k_x q_-=0\,,
	\end{equation}
	and thus to 
	\begin{align}
		X_- &= k_x^2+\kappa_{ev}^2=k_x^2(1-q_-^2)\nonumber\\
		&= \frac{k_x^2}{\omega^2}\left(\omega^2-(f-\nu X_-)^2\right)=\frac{k_x^2}{\omega^2}X_-
	\end{align}
	or equivalently $(\omega^2-k_x^2)X_-=0$. Since both factors are known to be non-zero this is impossible.
\end{rem}

\section{Dirichlet boundary condition\label{app:Dirichlet}}

The computation of the scattering amplitude is similar to the one in Sect.~\ref{sec:BSC_proof}, but with \eqref{eq:boundary_condition} replaced by \eqref{eq:Dirichlet_BC} we end up with the simpler form 
\begin{equation}
S_\zeta(k_x,\kappa) = - \dfrac{g_\zeta(k_x,-\kappa)}{g_\zeta(k_x,\kappa)},\qquad g_\zeta(k_x,\kappa) = \begin{vmatrix}
 u_\zeta(k_x,\kappa)  & u_\infty(k_x,\kappa_\mathrm{ev})\\
v_\zeta(k_x,\kappa) & v_\infty(k_x,\kappa_\mathrm{ev})
\end{vmatrix}.
\end{equation}
In particular, for $\zeta =0$ one infers $g_0(k_x,\kappa) \to 2\ii$ as $(k_x,\kappa) \to \infty$, so that $S_0 \to -1$. Hence the scattering amplitude has no zero or pole in the neighborhood of $\infty$ and does not wind either. This is the regular situation where Levinson's theorem applies. Apart from \eqref{eq:Levinson_violation}, the rest of Thm.~\ref{thm:main} applies indeed similarly for Dirichlet boundary condition. This implies $C_+=n_\mathrm{b}$ so that the bulk-edge correspondence is satisfied.


\begin{thebibliography}{bib}
	
\bibitem{Avilaetal13}
Avila, J. C., Schulz-Baldes, H., and Villegas-Blas, C. (2013) \textit{Topological invariants of edge states for periodic two-dimensional models.} Mathematical Physics, Analysis and Geometry \textbf{16(2)} 137-170
	
\bibitem{Avron98}
Avron, J. E. (1998) \textit{Odd viscosity.} Journal of statistical physics \textbf{92(3-4)} 543-557

\bibitem{Bal19}
Bal, G. (2019) \textit{Continuous bulk and interface description of topological insulators.} Journal of Mathematical Physics \textbf{60(8)} 081506

\bibitem{BanerjeeSouslovAbanovVitelli17}
Banerjee, D., Souslov, A., Abanov, A. G., and Vitelli, V. (2017) \textit{Odd viscosity in chiral active fluids.} Nature communications \textbf{8(1)} 1573

\bibitem{DelplaceMarstonVenaille17}
Delplace, P., Marston, J. B., and Venaille, A. (2017) \textit{Topological origin of equatorial waves.} Science \textbf{358(6366)} 1075-1077

\bibitem{DeNittisLein17}
De Nittis, G., and Lein, M. (2017) \textit{Symmetry classification of topological photonic crystals.} arXiv preprint arXiv:1710.08104.

\bibitem{Drouot19}
Drouot, A. (2019) \textit{The bulk-edge correspondence for continuous honeycomb lattices.} arXiv preprint arXiv:1901.06281.

\bibitem{EssinGurarie11}
Essin, A.M., Gurarie, V. (2011) \textit{Bulk-boundary correspondence of topological insulators from their Green's functions.} Physical Review B \textbf{84} 125132

\bibitem{Feffermanetal16}
Fefferman, C. L., Lee-Thorp, J. P., and Weinstein, M. I. (2016)\textit{ Edge states in honeycomb structures.} Annals of PDE \textbf{2(2)} 12

\bibitem{FroehlichStuderThiran95}
Fr{\"o}hlich, J., Studer, U.M., and Thiran, E. (1995) \textit{Quantum theory of large systems of non-relativistic matter} arXiv preprint cond-mat/9508062

\bibitem{GrafPorta13}
Graf, G. M., and Porta, M. (2013) \textit{Bulk-edge correspondence for two-dimensional topological insulators.} Communications in Mathematical Physics \textbf{324(3)} 851-895

\bibitem{Halperin82}
Halperin, B. I. (1982) \textit{Quantized Hall conductance, current-carrying edge states, and the existence of extended states in a two-dimensional disordered potential.} Physical Review B \textbf{25(4)} 2185

\bibitem{Hatsugai93}
Hatsugai, Y. (1993) \textit{Chern number and edge states in the integer quantum Hall effect.} Physical Review Letters \textbf{71(22)} 3697

\bibitem{Iga95}
Iga, K. (1995) \textit{Transition modes of rotating shallow water waves in a channel.} Journal of Fluid Mechanics \textbf{294} 367-390

\bibitem{KellendonkPankrashkinRichard11}
Kellendonk, J., Pankrashkin, K., and Richard, S. (2011) \textit{Levinson's theorem and higher degree traces for Aharonov-Bohm operators.}  Journal of Mathematical Physics \textbf{52(5)}

\bibitem{ProdanSchulzBaldes16}
Prodan, E., and Schulz-Baldes, H. (2016) \textit{Bulk and boundary invariants for complex topological insulators. From K-theory to physics} Mathematica Physics Studies, Springer

\bibitem{Perietal19}
Peri, V., Serra-Garcia, M., Ilan, R., and Huber, S. D. (2019) \textit{Axial-field-induced chiral channels in an acoustic Weyl system.} Nature Physics \textbf{15(4)} 357

\bibitem{RaghuHaldane08}
Raghu, S., and Haldane, F. D. M. (2008) \textit{Analogs of quantum-Hall-effect edge states in photonic crystals.} Physical Review A \textbf{78(3)} 033834

\bibitem{SchulzBaldesKellendonkRichter00}
Schulz-Baldes, H., Kellendonk, J., Richter, T. (2000) \textit{Simultaneous quantization of edge and bulk Hall conductivity.}  J. Phys. A: Math. Gen. 33, L27

\bibitem{Souslovetal19}
Souslov, A., Dasbiswas, K., Fruchart, M., Vaikuntanathan, S., and Vitelli, V. (2019) \textit{Topological waves in fluids with odd viscosity.} Physical Review Letters \textbf{122(12)} 128001

\bibitem{TauberDelplaceVenaille19}
Tauber, C., Delplace, P., and Venaille, A. (2019) \textit{A bulk-interface correspondence for equatorial waves.} Journal of Fluid Mechanics \textbf{868}

\bibitem{TauberDelplaceVenaille19bis}
Tauber, C., Delplace, P., and Venaille, A. (2019) \textit{Anomalous bulk-edge correspondence in continuous media.} arXiv:1902.10050

\bibitem{Vallis17}
Vallis, G. (2017) \textit{Atmospheric and Oceanic Fluid Dynamics: Fundamentals and Large-Scale Circulation.}  Cambridge University Press

 

\end{thebibliography}
\end{document}